\documentclass{article}
\usepackage[utf8]{inputenc}
\usepackage[T1]{fontenc}
\usepackage[english]{babel}
\usepackage[ruled,lined,noend]{algorithm2e}
\usepackage[style=alphabetic]{biblatex}
\usepackage{
    lmodern,
    amsmath,
    amsfonts,
    amssymb,
    amsthm,
    hyperref,
    cleveref,
    thmtools,
    xcolor,
    csquotes,
    mathtools,
    xspace,
    enumitem,
    booktabs,
    array,
    longtable,
    geometry,
    dsfont,
}

\declaretheorem[name=Theorem]{theorem}
\declaretheorem[name=Lemma, sibling=theorem]{lemma}
\declaretheorem[name=Proposition, sibling=theorem]{proposition}
\declaretheorem[name=Corollary, sibling=theorem]{corollary}
\declaretheorem[name=Definition, sibling=theorem]{definition}
\declaretheorem[name=Observation, sibling=theorem]{observation}
\declaretheorem[name=Example, sibling=theorem, style=remark]{example}
\declaretheorem[name=Remark, sibling=theorem, style=remark]{remark}

\crefname{appendix}{appendix}{appendices}
\Crefname{appendix}{Appendix}{Appendices}

\SetArgSty{textup}

\SetCommentSty{algcommentfont}
\SetKwComment{Comment}{\textcolor{black!50}{//}\ }{}
\SetKwProg{Upon}{upon}{:}{}
\SetKw{KwAnd}{and}
\SetKwProg{Fn}{procedure}{:}{}
\SetKwFunction{FCrit}{CriticalService}
\SetKwFunction{FPre}{PreemptiveService}
\SetFuncSty{textsc}
\SetInd{0.3em}{0.7em}
\SetKwIF{If}{ElseIf}{Else}{if}{}{else if}{else}{end if}
\SetKwIF{IfThen}{ElseIfThen}{ElseThen}{if}{then}{else if}{else}{end if}
\SetKwFor{For}{for}{:}{end for}
\SetKwFor{While}{while}{:}{end while}

\DontPrintSemicolon

\newcommand{\OPT}{\textnormal{\textsc{Opt}}\xspace}
\newcommand{\PMD}{\textnormal{\textsc{PMD}}\xspace}
\newcommand{\ALG}{\textnormal{\textsc{Alg}}\xspace}
\newcommand{\EV}{\mathrm{EV}}
\newcommand{\Exc}{\mathrm{Exc}}
\newcommand{\fetch}{\mathrm{fetch}}
\newcommand{\dd}{\mathrm d}

\newcommand{\dur}[1]{\lvert #1\rvert}
\newcommand{\durW}[2]{\lvert #2\rvert_{#1}}
\newcommand{\send}{\dur{\lambda}}
\newcommand{\lpath}[2]{\lambda_{#1\to #2}}
\newcommand{\lightstep}[1]{\nu_{j,W,#1}}
\newcommand{\Crit}{S_{\mathrm{crit}}}
\newcommand{\CritAt}[1]{S_{\mathrm{crit}}(#1)}
\newcommand{\Hit}{S_{\mathrm{hit}}}
\newcommand{\Miss}{S_{\mathrm{miss}}}
\newcommand{\Over}{S_{\mathrm{over}}}
\newcommand{\Timely}{S_{\mathrm{timely}}}
\newcommand{\Visit}{S_{\mathrm{visit}}}
\newcommand{\crit}{\mathrm{crit}}
\newcommand{\hit}{\mathrm{hit}}
\newcommand{\miss}{\mathrm{miss}}
\newcommand{\deadline}{\mathrm{deadline}}
\newcommand{\delayof}{\mathrm{delay}}

\renewcommand{\O}{\mathcal O}

\usepackage{tikz}
\usetikzlibrary{arrows.meta, decorations.pathreplacing}
\definecolor{figpink}{RGB}{251,64,133}
\definecolor{figgreen}{RGB}{0,153,136}
\definecolor{figblue}{RGB}{0,119,187}
\definecolor{figorange}{RGB}{238,119,51}

\title{Tight weighted paging with delay via mirror descent}
\author{Julian Born\thanks{Technical University Darmstadt. Email: born@mathematik.tu-darmstadt.de. Supported by DFG grant DI~2041/2-2.} \and Yann Disser\thanks{Technical University Darmstadt. Email: disser@mathematik.tu-darmstadt.de} \and Linda Thelen\thanks{Technical University Darmstadt. Email: thelen@mathematik.tu-darmstadt.de. Supported by a doctoral scholarship from the Studienstiftung des deutschen Volkes.}}

\date{}

\begin{document}
\maketitle

\begin{abstract}
    We give an~$\O(\log k)$-competitive randomized algorithm for weighted paging with delay.
    This improves upon the~$\O(\log k\log n)$-competitive algorithm of Gupta, Kumar, and Panigrahi (STOC 2020, SICOMP 2022) by removing the dependence on the number~$n$ of pages.
    Without delay, the optimal competitive ratio of randomized algorithms is~$\Theta(\log k)$, and since the lower bound carries over, our bound is tight.
    Our algorithm is fractional and combines mirror descent, used by Bubeck et al.\ (STOC 2018) for weighted paging without delay, with preemptive service, introduced by Azar et al.\ (STOC 2017) for online service with delay:
    when a request is served, an increase in the mirror descent potential is offset by preemptive service to other pending requests, with a budget derived from the mirror descent path.
    To obtain a randomized algorithm, we use the online rounding of Bansal, Buchbinder, and Naor (FOCS 2007) from classical paging.
\end{abstract}

\section{Introduction}

Paging is one of the most fundamental problems in online computation.
In weighted paging, a cache holds up to~$k$ of~$n$ pages, and a requested page that is not in the cache must be fetched, at a cost equal to its weight, while another page is evicted.
Classical paging requires every request to be served immediately.
In many applications of caching, however, requests can wait for some time, for example in prefetching, cloud computing, or video-on-demand streaming.
An algorithm can then save cost by batching requests to the same page instead of fetching it repeatedly.
In \emph{paging with deadlines}, waiting is free, but every request must be served no later than its deadline.
More generally, in \emph{paging with delay}, every request accrues delay cost until it is served, and the goal is to minimize the total fetch cost plus delay cost.
Paging with delay belongs to an active line of work on online problems with delay, including TCP acknowledgment \cite{Dooly01TCPAck}, matching \cite{HasteMakesWaste}, online service with delay \cite{Azar21OSD}, multi-level aggregation \cite{Bienkowski16OnlineMLA,Buchbinder17OdMLA}, and network design \cite{Azar19MetricOpt,Azar20NDD}.

When all pages have the same weight, delay does not make paging harder, since unweighted paging with delay reduces to classical paging \cite{Azar21OSD}.
For general weights, however, the best known randomized algorithm, due to \citeauthor{Gupta22CachingWithTimeWindows}, is~$\O(\log k\log n)$-competitive \cite{Gupta22CachingWithTimeWindows}, where the number~$n$ of pages can be much larger than~$k$.
Without delay, the optimal competitive ratio of randomized algorithms is~$\Theta(\log k)$ \cite{Fiat91CompetitivePaging,Bansal07PrimalDual}, and the lower bound carries over to paging with delay.
The factor~$\log n$ comes from one step of the LP-based approach of \citeauthor{Gupta22CachingWithTimeWindows}.
The natural relaxation over fractional cache states has integrality gap~$\Omega(k)$, and \citeauthor{Gupta22CachingWithTimeWindows} circumvent this with a hitting-set relaxation that records when pages are fetched and evicted, but not the cache contents.
This requires a separate step that decides online which pending requests to serve and when, which costs only a constant factor offline, but a factor~$\log n$ online.
They leave closing the gap to the lower bound as an open problem \cite{Gupta22CachingWithTimeWindows}.

We give an~$\O(\log k)$-competitive randomized algorithm for weighted paging with delay, in the same clairvoyant model as \cite{Gupta22CachingWithTimeWindows}.
This matches the lower bound, does not depend on~$n$, and also applies to deadlines.
Hence, weighted paging with delay is as hard as without delay, up to constant factors.
For deterministic algorithms, this was already known, with optimal competitive ratio~$\Theta(k)$ \cite{Sleator85AmortizedEfficiency,Chrobak91kServerTrees,Azar21OSD}.
We establish the same equivalence for randomized algorithms, where the optimal competitive ratio is~$\Theta(\log k)$.

Our algorithm combines \emph{mirror descent}, used by \citeauthor{Bubeck18kServerMultiscale} for classical weighted paging \cite{Bubeck18kServerMultiscale}, with \emph{preemptive service}, a standard tool for online service with delay \cite{Azar21OSD}.
With delay, the potential function in the analysis of mirror descent can increase when we serve a request, because \OPT may serve the same request at another time than we do.
We compensate for this increase by preemptive service: when serving a request, we also serve other pending requests, with a budget derived from the mirror descent path.
To our knowledge, this is the first use of mirror descent for a problem with delay.
Since the fractional cache states of our algorithm also determine when requests are served, no separate step has to decide this, and the classical online rounding for weighted paging \cite{Bansal07PrimalDual} applies with only a constant loss.

We state our result formally below, and outline our technique in \Cref{sec:technique}.

\subsection{Our results}

We give a randomized algorithm for weighted paging with delay with tight competitive ratio that, in particular, does not depend on the number~$n$ of pages.

\begin{theorem}\label{thm:main}
    There exists a clairvoyant randomized online algorithm for weighted paging with delay with competitive ratio~$\O(\log k)$.
\end{theorem}

Together with the lower bound of~$\Omega(\log k)$ from classical paging \cite{Fiat91CompetitivePaging}, this gives a tight bound on the competitive ratio of randomized weighted paging with delay and deadlines.
We obtain this result by combining the \emph{mirror descent} technique from \cite{Bubeck18kServerMultiscale} for classical paging with a \emph{preemptive service} approach typical for online service with delay \cite{Azar21OSD}.
The resulting algorithm \PMD is explicit (Algorithm~\ref{alg:PMD}) and handles delay directly, with three rules: when to serve (once the delay of a page is large compared to the cost of serving it), how to serve (along a mirror descent path in closed form), and what else to serve (pending requests in order of urgency, paid from budgets).
\Cref{thm:main} is verified in Lean 4, using only the online rounding procedure of \citeauthor{Bansal07PrimalDual} as an assumption.\footnote{%
    The Lean formalization is available at \url{https://github.com/licsth/paging-with-delay-pmd-formalization}.
    It follows our proof tree exactly, with three minor exceptions.
    First, the Lean formalization resolves ties in request timing according to a fixed order on the requests to give a concrete algorithm definition.
    For our analysis, any tiebreaking rule suffices, and we do not specify one.
    Second, for presentation, some loops in our algorithm iterate over all possible weights, even though only a finite number of iterations trigger some action (see the proof of \Cref{lem:inv}).
    The formalization gives the finite ground set directly.
    Third, the Lean formalization uses the exact constants instead of \(\O\)-notation.
}

\paragraph{Significance}
To our knowledge, our algorithm and its analysis is the first use of mirror descent for any online problem with delay.
Our analysis suggests an approach for applying mirror descent to other problems with delay, such as online service with delay on HSTs, where mirror descent gives the best known bounds in the setting without delay \cite{Bubeck18kServerMultiscale}.

\paragraph{Concurrent work}
In independent work uploaded to arXiv less than a week after ours, \citeauthor{Naor26TimeWindows} \cite{Naor26TimeWindows} also present an~$\O(\log k)$-competitive randomized algorithm for weighted paging with delay, using a fundamentally different approach from ours.
Their algorithm reuses the online cache configurations of \cite{Gupta22CachingWithTimeWindows} and reduces the decision which requests to serve to joint replenishment with time-dependent opening costs, for which they give a constant-competitive algorithm.
In contrast, we integrate mirror descent with preemptive service in a single randomized algorithm whose fractional cache states determines when requests are served.

\subsection{Technique}\label{sec:technique}
\paragraph{Fractional model.}
Instead of defining a randomized algorithm for paging with delay that makes randomized decisions about which pages to fetch and evict, we design a \emph{fractional} algorithm.
This approach was first introduced by \citeauthor{Bansal07PrimalDual} for weighted paging \cite{Bansal07PrimalDual} and later reused by \citeauthor{Gupta22CachingWithTimeWindows} for paging with delay \cite{Gupta22CachingWithTimeWindows}.
A fractional algorithm maintains, at every time, a fractional state~$x \in [0,1]^n$, where the \emph{antimass}~$x_i$ represents the fraction of page~$i$ that is \emph{not} cached.
Equivalently,~$1-x_i$ is the fractional \emph{cached mass} the algorithm has on page~$i$.
Throughout the execution, we maintain that the fractional state satisfies~$\sum_i x_i \geq n-k$, meaning that the total cached mass is at most~$k$.
\citeauthor{Bansal07PrimalDual} showed for classical paging that the state~$x$ of a fractional algorithm can be transformed online into a randomized integral paging algorithm realizing the state as its marginal probabilities, i.e., the probability of page~\(i\) being cached is~\(1-x_i\) \cite[Theorem 4.1]{Bansal07PrimalDual}.
In the design of our algorithm, we ensure that a requested page~$i$ has~$x_i=0$ at some time, so that the resulting randomized algorithm caches~$i$ almost surely at this time.

\paragraph{Mirror descent in classical paging.}
\citeauthor{Bubeck18kServerMultiscale} \cite{Bubeck18kServerMultiscale} demonstrated how a fractional solution for classical paging (without delay) can be defined via a \emph{mirror descent path}, recovering the tight~$\O(\log k)$ bound already proven in \cite{Bansal07PrimalDual}.
Upon a request to an uncached page~$p$, their algorithm redistributes cached mass from other pages to~$p$ by following a path~$x$ through fractional states defined by the differential inclusion
\begin{equation}\label{eq:md}
    x'(t) \in \nabla^2 \Phi(x(t))^{-1}\cdot \left(-\mathbf{1}_p-N_S(x(t))\right).
\end{equation}
Here,~$\mathbf{1}_p$ is the vector with entry~$1$ at~$p$ and~$0$ at all other pages,~$S$ denotes the set of possible states~$x\in [0,1]^n$ with~$\sum_i x_i=n-k$,~$N_S(x)$ is the normal cone of~$S$ at~$x$, and~$\Phi$ is a strongly convex function called the \emph{mirror map}.
Intuitively,~$x'$ is defined to reduce the antimass at~$p$ (via the term~$-\mathbf{1}_p$), while being forced into the possible state space~$S$ (via the term~$-N_S(x(t))$), and following some distortion defined by~$\Phi$.
It can be shown that this mirror descent path indeed reduces the antimass at~$p$, while increasing the antimass at all other pages. The algorithm by \citeauthor{Bubeck18kServerMultiscale}~\cite{Bubeck18kServerMultiscale} follows this path until~$p$ is cached.
Upon each request, a new mirror descent path is followed and the resulting paths are concatenated to give a full execution of the algorithm.

To analyze this algorithm, \citeauthor{Bubeck18kServerMultiscale}~\cite{Bubeck18kServerMultiscale} use the \emph{Bregman divergence}
\[
    D(x,y) \coloneq \Phi(y)-\Phi(x)-\langle \nabla \Phi(x), y-x\rangle
\]
as a potential function.
\Cref{fig:md-potential} shows the potential along a mirror descent path satisfying \eqref{eq:md}.
Using the weighted entropy function~$\Phi=\sum_i w_i x_i \log(x_i)$ as the mirror map, the analysis of mirror descent paths with the Bregman potential from~\cite{Bubeck18kServerMultiscale} shows that
\[
    \int \mathbf{1}_p^\top x(t)\dd t \leq L \int \|\OPT'(t)\|_{\ell_1(w)}\dd t+\int \mathbf 1_p^\top \OPT(t)\dd t.
\]
Here,~$\ell_1(w)$ is the weighted~$\ell_1$-norm,~$\|x\|_{\ell_1(w)}=\sum_i w_i|x_i|$, so that, for a given differentiable path~$y$ of fractional states,~$\int \|y'(t)\|_{\ell_1(w)}\dd t$ is the cost of transitioning through the fractional states along~$y$ when moving mass into or out of page~$i$ costs~$w_i$, and~$L$ is the Lipschitz constant of~$\Phi$ over the state space~$S$.
The terms~$\int \mathbf{1}_p^\top x(t)\dd t$ and~$\int \mathbf 1_p^\top \OPT(t)\dd t$ are called the \emph{service cost} of the algorithm and the optimum, respectively, and~$\int \|\OPT'(t)\|_{\ell_1(w)}\dd t$ is the \emph{movement cost} of the optimum.
Writing~$x'_-\coloneqq \max\{0, -x'\}$ for the negative part of the derivative of~$x$, \citeauthor{Bubeck18kServerMultiscale}~\cite{Bubeck18kServerMultiscale} moreover show that the negative movement cost~$\int \|x'_-(t)\|_{\ell_1(w)}\dd t$ of the algorithm is bounded by its service cost, yielding
\begin{equation}\label{eq:bregman_guarantee}
    \int \|x'_-(t)\|_{\ell_1(w)}\dd t \leq L \int \|\OPT'(t)\|_{\ell_1(w)}\dd t+\int \mathbf 1_p^\top \OPT(t)\dd t
\end{equation}
for each request to some page~$p$.

In classical paging, every algorithm has to serve each requested page immediately.
Therefore, when page~$p$ is requested, it holds that~$\OPT_p(t)=0$, and hence the optimum's service cost is zero.
(Note that this is not the case for paging with delay, as discussed below.)
The competitive guarantee of~$L$ then follows immediately from \eqref{eq:bregman_guarantee}.
On the full state space~$S$, the Lipschitz constant~$L$ of the weighted entropy turns out to be unbounded.
In particular, on a domain~$[\delta,1]^n\subseteq [0,1]^n$, the Lipschitz constant is~$\Theta(-\log \delta)$.
To get a better competitive guarantee, \citeauthor{Bubeck18kServerMultiscale}~\cite{Bubeck18kServerMultiscale} therefore follow the mirror descent path only until~$x_p=\delta$ for some small threshold~$\delta$, at which point the page is considered served.
Finally, they apply a fixed transformation that maps~$\delta$ to zero, so that the requested page is truly served when~$x_p=\delta$.
The competitive ratio then is~$\O(-\log \delta)$.

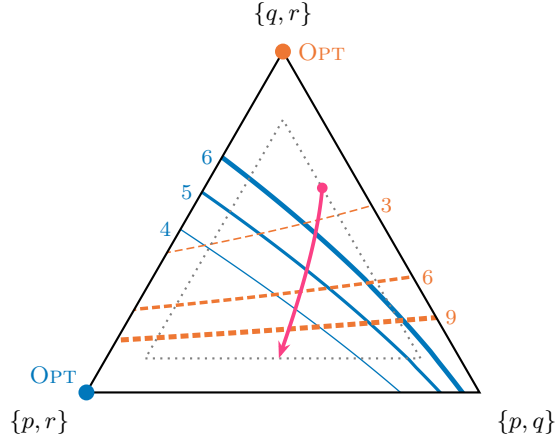
\begin{figure}[tbp]
    \centering
    \begin{tikzpicture}[x=1.3cm, y=1.3cm, color=black, font=\small]
        \begin{scope}
            \clip (0,0)--(4,0)--(2,3.464)--cycle;
            \foreach \c/\a/\b/\w in {4/0.199/0.524/0.5, 5/0.099/0.414/1.2, 6/0.043/0.312/1.8}
            \draw[figblue, line width=\w pt] plot[domain=\a:\b, samples=50, variable=\t]
            ({4*(1-\t)-2*((\c-(\t-1-ln(\t))-4*(1-\t))/4)}, {3.464*((\c-(\t-1-ln(\t))-4*(1-\t))/4)});
            \foreach \c/\a/\b/\w in {3/0.408/0.550/0.6, 6/0.242/0.342/1.3, 9/0.154/0.222/1.9}
            \draw[figorange, dash pattern=on 3pt off 1.5pt, line width=\w pt] plot[domain=\a:\b, samples=50, variable=\t]
            ({4*(1-\t-(8*(\t-1-ln(\t))+4*(1-\t)-\c)/3)+2*\t}, {3.464*\t});
        \end{scope}
        \foreach \c/\x/\y in {4/0.957/1.657, 5/1.176/2.036, 6/1.381/2.389}
        \node[figblue, left, font=\footnotesize] at (\x,\y) {$\c$};
        \foreach \c/\x/\y in {3/2.902/1.899, 6/3.321/1.176, 9/3.559/0.761}
        \node[figorange, right, font=\footnotesize] at (\x,\y) {$\c$};
        \draw[line width=0.8pt] (0,0)--(4,0)--(2,3.464)--cycle;
        \node[below left=3pt] at (0,0) {$\{p,r\}$};
        \node[below right=3pt] at (4,0) {$\{p,q\}$};
        \node[above=6pt] at (2,3.464) {$\{q,r\}$};
        \draw[black!50, dotted, line width=0.8pt] (0.6,0.346)--(3.4,0.346)--(2,2.771)--cycle;
        \fill[figblue] (0,0) circle (3pt);
        \node[figblue, above left] at (0,0) {\OPT};
        \fill[figorange] (2,3.464) circle (3pt);
        \node[figorange, right=2pt] at (2,3.464) {\OPT};
        \draw[-{Stealth[length=2.2mm]}, figpink, line width=1.4pt] plot[domain=0:1.5272, samples=40, variable=\s]
        ({1.2*exp(\s/4)+2*(1-0.1*exp(\s)-0.3*exp(\s/4))}, {3.464*(1-0.1*exp(\s)-0.3*exp(\s/4))});
        \fill[figpink] (2.4,2.078) circle (2pt);
    \end{tikzpicture}
    \caption{%
        Fetching page~$p$ via a mirror descent path.
        Here, there are~$n=3$ pages $\{p,q,r\}$ with weights~$(w_p,w_q,w_r)=(8,1,4)$, and the cache has size~$k=2$.
        The fractional states form a simplex whose vertices are the integral states, labelled by their caches.
        Our algorithm keeps its state in~$[\delta,1]^3$ (dotted) and follows the mirror descent path (pink) until~$x_p=\delta$.
        The lines are level sets of the Bregman potential for two states of \OPT.
        If \OPT caches~$p$ (blue, solid), the potential decreases along the path by at least the fetch cost.
        In paging with delay, \OPT need not cache~$p$ (orange, dashed), and the same path may increase the potential.
    }
    \label{fig:md-potential}
\end{figure}

\paragraph{Adapting mirror descent for paging with delay.}
A natural approach to reusing the analysis of \citeauthor{Bubeck18kServerMultiscale}~\cite{Bubeck18kServerMultiscale} for paging with delay would be to act as if a page has to be served once it has accumulated large delay compared to the cost of serving it.
In contrast to classical paging, however, when fetching such a page, we cannot assume that~$\OPT$ has the requested page~$p$ cached at the same time that we fetch it.
\OPT may serve the same request earlier or later than we do.
In this case, \eqref{eq:bregman_guarantee} gives the competitive ratio~$L$ only with an additive term~$\int \mathbf 1_p^\top \OPT(t)\dd t$ per fetch (the optimum's service cost).
In the analysis leading to \eqref{eq:bregman_guarantee}, this additive term stems from an increase of the Bregman potential while fetching the page, which is illustrated in \Cref{fig:md-potential}.

This is a fundamental obstruction to reusing the analysis as is: the increase is real, as we move away from \OPT's state, but \OPT may pay nothing at this time, since it pays for~$p$ earlier or later with additional delay.
The key ingredient of our analysis therefore is to charge this potential increase to costs that \OPT incurs at other times.
In particular, we relate the optimum's service cost back to the algorithm's positive movement cost, proving that
\begin{equation}\label{eq:serviceToPartialMovement}
    \int 1_p^\top \OPT(t)\dd t \leq 3 \int \|x'_+(t)\|_{\ell_1(w)}\dd t,
\end{equation}
where~$x'_+\coloneqq \max\{0, x'\}$.
By itself, this does not yet yield a competitive guarantee but it shows that investing an additional cost of~$\int 1_p^\top \OPT(t)\dd t$ changes the algorithm's cost by a constant factor, thus not deteriorating the competitive guarantee.
In particular, if this additional investment is for actions whose cost can be charged to \OPT, we can cover the increase of the Bregman potential and remain~$\O(L)$-competitive.

We achieve this by adding \emph{preemptive service}:
besides the page whose delay triggers a fetch, we also serve other pending requests, as long as an internal budget covers the cost of doing so.
Preemptive service is a key component of many existing algorithms for online service with delay \cite{Bienkowski18OSDLine, Azar21OSD, Touitou23ImprovedDeterministic}.
Usually, preemptive service is performed for roughly the cost of the triggering service and in order of decreasing urgency.
A request triggering critical service has either arrived after the last critical service, in which case \OPT could not have batched the request with others and therefore paid for it in full, or it was already present during the last preemptive service, in which case it might not cause cost for \OPT.
In this case, since preemptive service is in order of decreasing urgency, all preemptively served requests were more urgent than the current triggering request, and the current triggering request did not fit within the budget to be preemptively served.
This certifies that either, the algorithm performed its entire budget worth of preemptive service, or the previous preemptive service did not fully use its budget because the current service is more costly than the previous one.
In the first case, the cost can be charged directly to \OPT, and in the second case, the cost of one service is charged to that of the next service.
This creates a chain of increasingly heavy pages and the competitive ratio then measures the depth of such charging chains.

In our case, most cost is already covered through the Bregman potential, and the quantity to be bounded is instead the excess term~$\int 1_p^\top \OPT(t)\dd t$.
When a page is fetched due to high delay cost, we measure the length of the mirror descent path it follows and declare this amount as a budget for preemptive service to pending requests.
This length is at least the excess term~$\int \mathbf 1_p^\top \OPT(t)\dd t \leq \int 1 \dd t$ and can be bounded as in \eqref{eq:serviceToPartialMovement}, altering the cost of the algorithm by at most a constant factor.
Unlike in previous work on online service with delay, the excess to be bounded here can be much larger than the cost of the service, and this cannot easily be predicted.
Therefore, we allow every service to trigger preemptive services of any cost, irrespective of its own cost.
Moreover, every service can deposit budget for high-cost preemptive services, and we keep track of these budgets over the span of the entire execution, instead of per service or per page.

As usual in online service with delay, we perform preemptive service to pending requests in order of decreasing urgency.
Here, urgency is defined by the time at which the pending delay on a page exceeds the cost of serving it.
If this time has already passed, either \OPT has served the pending request already, incurring the same fetch cost as we do, or \OPT serves it later and pays the same cost in delay.
Either way, we can show that \OPT's cost covers the cost of our preemptive service.
In cases where budget was deposited by a service that increased the Bregman potential, we can show that almost the entire preemptive budget is spent on visits to such urgent requests, so that the increase in the potential is covered by the cost of \OPT.

\paragraph{Online rounding.}
To transform the resulting sequence of fractional states into a feasible one with~\(x_i = 0\) for every requested page~\(i\) at some point, we apply a fixed map to our fractional states, analogously to the approach in \cite{Bubeck18kServerMultiscale}.
To maintain the competitive guarantee, this map may not distort the cost of the resulting state transitions by more than a constant factor.
The distortion can be measured by the norm of the transformation map's derivative.
To optimize the competitive guarantee, we choose~$\delta$ maximal such that we can define a transformation mapping mass~$\delta$ to~$0$ with bounded derivative.
The optimal choice turns out to be~$\delta\coloneq 1/(2k+1)$ (see also \Cref{app:rounding}), giving the mirror map Lipschitz factor~$L=\log \frac 1 \delta=\log(2k+1)$ and competitive guarantee~$\O(\log(k))$.

\subsection{Related work}
Paging with delay can be seen as a special case of the \emph{online service with delay} (OSD) problem, where~$k$ agents serve requests in a metric space, which arrive over time and accrue delay cost while waiting to be served.
Weighted paging with delay can be understood as OSD on a star metric, and weighted paging with deadlines is the special case where delay costs remain zero up to the deadline, at which point they start growing rapidly.
A ray of the star then corresponds to a page with weight proportional to the ray's length, and the~$k$ agents represent the~$k$ cache slots.

Online service with delay was introduced by \citeauthor{Azar21OSD} in 2017 \cite{Azar21OSD} together with a randomized~$\O(k\log^5 n)$-competitive clairvoyant algorithm for general metric spaces of size~$n$, which is based on a deterministic~$\O(kd^4)$-competitive clairvoyant algorithm for hierarchically separated trees (\emph{HSTs}) of depth~$d$, combined with a randomized embedding of general metrics into HSTs.
In particular, since any star metric is an HST of depth~$1$, this implies a clairvoyant deterministic~$\O(k)$-competitive algorithm for paging with delay or deadlines.
For unweighted stars, \citeauthor{Azar21OSD} give a simple reduction to classical paging, providing an~$\O(\log k)$ bound for unweighted paging with delay.
The authors complement this with a lower bound of~$\Omega(\sqrt n)$ for the competitive ratio of deterministic nonclairvoyant weighted paging with deadlines with in the case of~$k=1$.
In \cite{Krnetic20kServerUniform} and \cite{Melnyk26TwoHoles}, the competitive ratio of both clairvoyant and nonclairvoyant unweighted paging with delay has since been shown to lie in~$[2k+1,2k+2]$.

For online service with delay on general metrics, the upper bound for the competitive ratio of deterministic clairvoyant algorithms has been reduced to~$\O(\log n)$ in the case of~$k=1$ \cite{Touitou23ImprovedDeterministic}, while no better bounds have been proven for general~$k$.
In the nonclairvoyant case, a deterministic upper bound of~$\O(\sqrt n \log n)$ can be achieved~\cite{Touitou25NonclairvoyantOSD}.
For randomized algorithms, \citeauthor{Gupta22HittingSet} \cite{Gupta22HittingSet} gave an improved dependence on~$k$ through an~$\O(\mathrm{poly}\log(n,\Delta))$-competitive algorithm, where~$n\geq k$ is the size of the metric space and~$\Delta$ is the metric's aspect ratio~$\max_{x,y}d(x,y)/\min_{x\neq y}d(x,y)$.
In \cite{Gupta22CachingWithTimeWindows}, \citeauthor{Gupta22CachingWithTimeWindows} present an~$\O(\log k \log n)$-competitive randomized clairvoyant algorithm for paging with delay or deadlines through a primal-dual approach.
We improve the competitive ratio to~$\O(\log k)$ for this setting.

We can recover the classical paging problem as a special case of paging with delay by giving all requests immediate deadlines.
In particular, this implies that the lower bounds of~$k$ for deterministic algorithms \cite{Sleator85AmortizedEfficiency} and~$H_k\approx \log k$ for randomized algorithms~\cite{Fiat91CompetitivePaging} transfer directly to both clairvoyant and nonclairvoyant paging with delay, leading to a tight analysis of deterministic weighted paging with delay and a~$\Theta(\log n)$ gap between the known randomized upper and lower bounds, which we close in this paper.

Our algorithm uses the mirror descent technique, which was initially introduced in the context of convex optimization \cite{Nemirovskij83ProblemComplexity} and has recently seen many applications in randomized online decision making \cite{Buchbinder14CompetitiveRegularization, Bubeck21MTSMirrorDescent, Bai25PassingACrown}.
Notably, \citeauthor{Bubeck18kServerMultiscale}~\cite{Bubeck18kServerMultiscale} used the continuous-time mirror descent framework to derive an~$\O(\log^2 k)$-competitive randomized algorithm for the~$k$-server problem on HSTs.
In their paper, they demonstrate how mirror descent can be used to recover the~$\O(\log k)$ guarantee for randomized online paging already achieved through the primal-dual analysis in \cite{Bansal07PrimalDual}.

\section{Problem definition}\label{sec:problem}
An instance of \emph{weighted paging with delay} is given by a finite sequence of requests~$r=(p_r, a_r, d_r)$ where~$p_r \in \{1,\dots,n\}$ is the requested page,~$a_r\ge0$ is the request's arrival time, and~$d_r\colon[0,\infty)\to[0,\infty)$ is its delay cost function.
Delay cost functions are nondecreasing, continuous and unbounded, with~$d_r(0)=0$.
We consider the \emph{clairvoyant} version of the problem where delay functions are revealed upon the arrival of a request.

An algorithm holds a set of up to~$k$ pages at any given time in its \emph{cache}, and every algorithm, as well as the optimum, starts with a fixed initial cache of~$k$ pages.
The algorithm may alter its cache instantaneously at any time, and may go through multiple cache configurations at one instant of time.
Each page~$i$ is associated with a weight~$w_i$ and fetching a page incurs a cost equal to its weight.
The eviction of pages does not incur any cost.
A request~$r$ is \emph{served at time~$\tau$} if~$\tau\ge a_r$ is the first time such that~$p_r$ is cached at time~$\tau$ (i.e., it is part of some cache configuration that the algorithm goes through at time~$\tau$).
The request then incurs \emph{delay cost}~$d_r(\tau-a_r)$.
The cost of an algorithm is its total fetch cost plus its total delay cost.

An algorithm is \emph{online} if its actions at time~$t$ depend only on the requests with arrival time at most~$t$.
Furthermore, if its actions at time~$t$ do not depend on the trajectory of the delay cost curves beyond~$t$, the algorithm is \emph{nonclairvoyant}, otherwise it is \emph{clairvoyant}.

An algorithm is~\emph{$\rho$-competitive} if, for every input~$\sigma$, its cost satisfies~\(
\ALG(\sigma)\leq \rho\cdot \OPT(\sigma) + c,
\)
where~$\OPT(\sigma)$ is the cost of the optimal offline algorithm on~$\sigma$, and~$c$ is a constant that is independent of the input sequence.
A randomized algorithm is~\(\rho\)-competitive if its expected cost satisfies~\(\mathbb E[\ALG(\sigma)] \le \rho\cdot\OPT(\sigma)+c\), where the expectation is taken over the algorithm's random choices.

For the analysis we use the following convention of order.
At a time~$t$, first the requests with~$a_r=t$ arrive, then the optimal offline schedule performs its cache updates and finally the online algorithm updates its cache.
``Before'' and ``after'' refer to this order.

Furthermore, we assume w.l.o.g.\ that the optimum schedule is lazy in the sense that every eviction is delayed until the next fetch, so that the cache is full at all times, and every eviction is immediately followed by a fetch.
This does not increase its delay or fetch cost:
The lazy schedule caches a superset of the pages of the original schedule at all times and fetches every page at the same time as the original schedule, or not at all.

\section{Algorithm definition}\label{sec:alg}
In this section, we define our algorithm \PMD (\emph{Preemptive Mirror Descent}) for weighted paging with delay.
The pseudocode of \PMD is given in Algorithm \ref{alg:PMD}.
For the definition of \PMD, we assume that all weights are powers of two (see the proof of \Cref{thm:main}).
Moreover, we assume w.l.o.g.\ that the minimum page weight is~$1$.
Throughout, we fix the threshold~$\delta\coloneqq\frac1{2k+1}$ and let~$L\coloneqq\log\frac1\delta=\log(2k+1)$, where~$\log$ is the natural logarithm.

\PMD keeps a fractional state~$x\in[\delta,1]^n$, starting with the initial cache.\footnote{
    Technically, the initial cache does not have~\(x_i \ge \delta\) for the cached pages, so \PMD immediately moves to a valid initial state~\(x\in[\delta,1]^n\).
    Since this is a one time movement of cost at most~\(\sum_i w_i\), this is irrelevant to the competitive ratio.
}
A request~$r$ for page~$p$ is \emph{served} at the first time~$\tau_r\ge a_r$ at which~$x_p=\delta$, and it is \emph{pending} from its arrival until it is served.
Because we define the randomized algorithm after applying a transformation that maps~$\delta$ to~$0$ (\Cref{prop:rounding}), every realization of the randomized algorithm has~\(p\) cached with probability~$1$ at time~$\tau_r$.
In case a page's antimass is strictly positive, it may be cached even before this time in some realizations of the algorithm.

\PMD remains idle until the delay of pending requests pays for serving some page.
We then call this page \emph{critical}, and two types of service are performed.
First, the critical page is served during a \emph{critical service}, which is followed by \emph{preemptive services} to other pending requests.

\begin{algorithm}[tbp]
    \caption{\textsc{Preemptive Mirror Descent} (\PMD)}
    \label{alg:PMD}
    \lFor{each $W\in\{1,2,4,\dots\}$ and page $p$}{$B_W\gets 0$, \quad $c_p\gets0$\Comment*[f]{initialize budgets and claims}}
    \Upon{page \(p\) becoming critical (i.e., last time before \eqref{eq:trigger} holds)}{
        \FCrit{$p$}\;
        \For{each level $W\in\{1,2,4,\dots\}$}{
            \If{$B_W \ge 4W$ \KwAnd $\sum_{i:w_i\le W}(1-x_i) \ge 1-\delta$}{
                \FPre{$W$}\;
            }
        }
    }
    \BlankLine
    \Fn{\FCrit{$p$}}{
        $\theta\gets(\dur{\lpath{x}{p}}/L)/(\sum_{w_q\le w_p}o_q)$\;
        \lFor{each page $q$ with $w_q\le w_p$}{$c_q\gets c_q+\theta o_q$\Comment*[f]{claim delay, total increase $\dur{\lpath{x}{p}}/L$}}
        \lFor{each level $W \in \{w_p,2w_p,\dots\}$}{$B_W \gets B_W + \durW{W}{\lpath{x}{p}}$\Comment*[f]{deposit budget}}
        $x\gets\lpath{x}{p}(\dur{\lpath{x}{p}})$, \quad $c_p\gets0$\Comment*[r]{serve $p$, reset claim}
    }
    \BlankLine
    \Fn{\FPre{$W$}}{
        $x^{\mathrm{start}}\gets x$\;
        $\theta\gets(1-\delta)/(\sum_{w_i\le W}(1-x_i))$\;
        \lFor{each page $i$ with $w_i\le W$}{$x_i\gets x_i+\theta(1-x_i)$ \Comment*[f]{open a hole, total increase $1-\delta$}}
        $B_W\gets B_W-W(1-\delta)$\Comment*[r]{pay for closing the hole}
        $E \gets \{\text{page }q\text{ with pending request}:w_q\leq W\}$ ordered by incr.~$\deadline(q)$\Comment*[r]{eligible pages $E$}
        \While{$E\neq \emptyset$ and $B_W\ge W$}{
            $q\gets E$.pop(0)\Comment*[r]{get and remove first page from $E$}
            $B_W\gets B_W-w_q(x_q-\delta)$\Comment*[r]{pay for the round trip}
            $x_q\gets\delta\gets x_q$, \quad $c_q\gets0$\Comment*[r]{set $x_q$ to $\delta$ and back (serve $q$), reset claim}
        }
        $x\gets x^{\mathrm{start}}$\Comment*[r]{close the hole}
    }
\end{algorithm}

\paragraph{Critical service}
When a page~$p$ becomes critical (see below), \PMD serves it by following a mirror descent path.
We call this part of the execution a \emph{critical service} to~$p$ and say that it has \emph{level}~$w_p$, where~$w_p$ is the weight of page~$p$.

To serve~$p$, \PMD decreases the antimass at~$p$ and increases it at all other pages.
It does so by virtually following the mirror descent path from \cite{Bubeck18kServerMultiscale} defined in \eqref{eq:md}.
We denote this path by~$\lambda$ and write~$\lpath{x}{p}$ when we need to make the critical page~$p$ and the start state~$x\in [\delta, 1]^n$ explicit.
To simplify the analysis, we specify the path~$\lambda=\lpath{x}{p}$ explicitly as
\[
    \lambda(s)_i\coloneqq\min\{1,\ x_ie^{s/w_i}\}\quad(i\neq p),
    \qquad \lambda(s)_p\coloneqq n-k-\sum_{i\ne p}\lambda(s)_i\qquad(s\ge0).
\]
\Cref{fig:md-coordinates} illustrates the coordinates of such a path.
\Cref{app:mdpath} shows that this path is indeed the mirror descent path.
\PMD transitions to the state~$\lambda(\send)$ with~$\send\coloneq \min\{t\geq 0:\lambda(t)_p=\delta\}$ chosen minimally such that~$\lambda(\send)_p = \delta$.
We refer to~$\send$ as the \emph{length} of the critical service, which is the time it takes to serve page~$p$ along the mirror descent path.
Note that such a time~$\send$ exists because the mirror descent path is continuous,~$\lambda(s)_p$ is nonincreasing,~$\lambda(0)_p=x_p>\delta$, and~$\lambda(s)_p\to 1-k\le0$ as~$s\to\infty$.
Therefore,~\(\dur{\lambda}\) is well defined and~\(\lpath{x}{p}\) serves all pending requests for~\(p\).

Intuitively, the mirror descent path evicts mass from all pages~$i\ne p$ at rate~$\lambda(s)_i/w_i$, i.e., proportional to their current antimasses divided by their weights.
This way, \PMD favors evicting light pages that hold a lot of antimass, intuitively following a cheap path towards a state where the antimass is divided equally among all pages.

During the critical service, only~$x_p$ decreases, and the total incurred fetch cost therefore is
\[
    \fetch(\lpath{x}{p})\coloneqq w_p(\lambda(0)_p-\lambda(\send)_p)=w_p(x_p-\delta).
\]

\begin{figure}[tbp]
    \centering
    \begin{tikzpicture}[x=1.75cm, y=4.2cm, color=black, font=\small]
        \fill[black!7] (1,-1) rectangle (2,0);
        \fill[figorange!30] (2,-1) rectangle (5.448,0);
        \fill[figorange!10] (5.448,-1) rectangle (6.6,0);
        \foreach \a/\b/\W in {1/2/1, 2/4/2, 4/6.6/4} {
        \draw[black!30] (\a,-1)--(\a,0.2);
        \node[above, align=center] at ({(\a+\b)/2},0) {$I_{\W}$\\[-1pt]{\scriptsize\color{black!60}only evicts}\\[-2pt]{\scriptsize\color{black!60}$i$ with $w_i>\W$}};
        }
        \foreach \w in {1,2,4} \draw[black!45, dotted, line width=0.8pt] (0,-1)--(\w,0);
        \node[black!60, font=\scriptsize, rotate=30, above] at (2.8,-0.3) {evicted by $s=Lw_i$};
        \foreach \w/\y in {1/-0.2216, 4/-0.1386, 16/-0.7909, 32/-0.8867}
        \draw[black!60, line width=0.8pt] plot[domain=0:5.448, samples=300] (\x, {min(0, \y+\x/\w)});
        \foreach \w/\y/\n in {1/-0.2216/a, 4/-0.1386/b, 16/-0.7909/c, 32/-0.8867/d}
        \node[black!60, font=\scriptsize, left] at (-0.04,\y) {$w_{\n}=\w$};
        \draw[figpink, line width=1.4pt] plot[smooth] coordinates {(0.000,-0.013) (0.020,-0.032) (0.040,-0.053) (0.060,-0.074) (0.080,-0.098) (0.100,-0.123) (0.120,-0.149) (0.140,-0.178) (0.160,-0.209) (0.180,-0.242) (0.200,-0.278) (0.220,-0.318) (0.222,-0.321) (0.240,-0.329) (0.260,-0.337) (0.280,-0.346) (0.300,-0.355) (0.320,-0.364) (0.340,-0.373) (0.360,-0.382) (0.380,-0.392) (0.400,-0.402) (0.420,-0.412) (0.440,-0.422) (0.460,-0.433) (0.480,-0.443) (0.500,-0.455) (0.520,-0.466) (0.540,-0.478) (0.555,-0.486) (0.560,-0.487) (0.580,-0.488) (0.600,-0.489) (0.700,-0.495) (0.800,-0.501) (0.900,-0.507) (1.000,-0.513) (1.100,-0.519) (1.200,-0.526) (1.300,-0.532) (1.400,-0.539) (1.500,-0.546) (1.600,-0.552) (1.700,-0.560) (1.800,-0.567) (1.900,-0.574) (2.000,-0.582) (2.100,-0.589) (2.200,-0.597) (2.300,-0.605) (2.400,-0.613) (2.500,-0.621) (2.600,-0.630) (2.700,-0.639) (2.800,-0.648) (2.900,-0.657) (3.000,-0.666) (3.100,-0.676) (3.200,-0.686) (3.300,-0.696) (3.400,-0.706) (3.500,-0.716) (3.600,-0.727) (3.700,-0.738) (3.800,-0.750) (3.900,-0.762) (4.000,-0.774) (4.100,-0.786) (4.200,-0.799) (4.300,-0.812) (4.400,-0.826) (4.500,-0.840) (4.600,-0.854) (4.700,-0.869) (4.800,-0.885) (4.900,-0.901) (5.000,-0.918) (5.100,-0.935) (5.200,-0.953) (5.300,-0.971) (5.400,-0.990) (5.448,-1.000)};
        \draw[figpink, dashed] (5.448,-1)--(5.448,0);
        \fill[figpink] (5.448,-1) circle (2pt);
        \draw[->] (0,-1)--(6.75,-1) node[right] {$s$};
        \draw[->] (0,-1)--(0,0.1) node[above] {antimass};
        \foreach \u/\l in {1/L, 2/2L, 4/4L} \draw (\u,-1)--(\u,-1.03) node[below] {$\l$};
        \node[below] at (5.448,-1.03) {$\dur{\lambda}$};
        \draw (0,0)--(-0.04,0) node[left] {$1$};
        \draw (0,-1)--(-0.04,-1) node[left] {$\delta$};
        \node[left, align=right] at (-0.05,-0.5) {$\lambda(s)_i$\\(log scale)};
        \draw[decorate, decoration={brace, amplitude=4pt, mirror}] (2,-1.17)--(3.98,-1.17)
        node[midway, below=4pt] {$\durW{2}{\lambda}=2L$ into $B_2$};
        \draw[decorate, decoration={brace, amplitude=4pt, mirror}] (4.02,-1.17)--(5.448,-1.17)
        node[midway, below=4pt] {$\durW{4}{\lambda}\approx 1.4L$ into $B_4$};
    \end{tikzpicture}
    \caption{%
        A critical service to a page~$p$ along the mirror descent path~$\lambda=\lpath{x}{p}$.
        Here, there are five pages~$\{p,a,b,c,d\}$ with weights~$(w_p,w_a,w_b,w_c,w_d)=(2,1,4,16,32)$.
        The plot shows the coordinates~$\lambda(s)_i$ over the parameter~$s$ on a logarithmic scale, coordinate $p$ is in pink.
        Every page~$i\ne p$ follows a line of slope~$1/w_i$ until it is fully evicted, which happens by~$s=Lw_i$ at the latest (see \eqref{eq:light_evicted_in_band}).
        The dotted lines show the slopes~$1,\frac12,\frac14$.
        Hence, during~$I_W$, only pages heavier than~$W$ are evicted.
        The path does not depend on $w_p$ and has kinks at each point that some page becomes fully evicted.
        \PMD deposits the parts $\durW{W}{\lambda}$ of $\lambda$ in~$I_W$ into the budgets~$B_W$ for all $W\geq w_p$.
    }
    \label{fig:md-coordinates}
\end{figure}
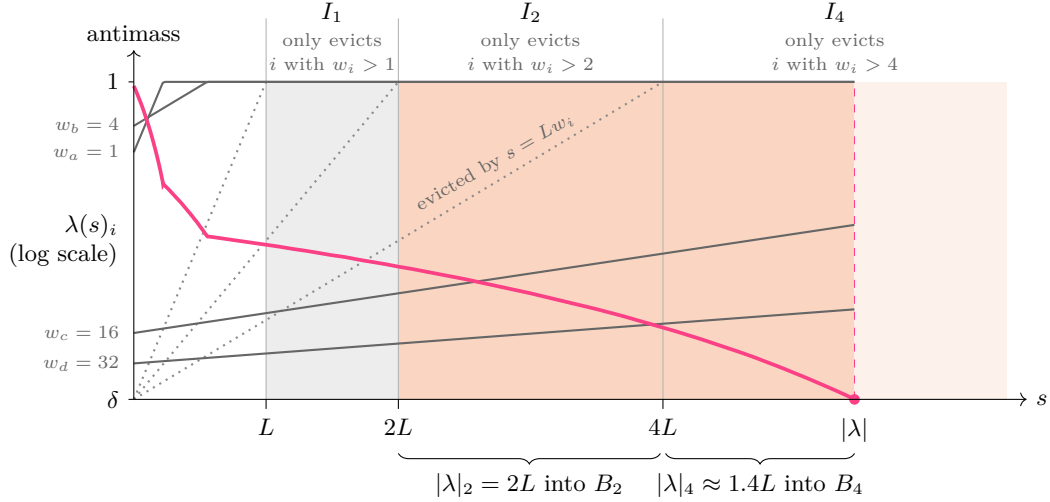

\paragraph{When to turn a page critical}

A natural rule would turn page~$p$ critical once the delay of its pending requests equals the cost~$w_p$ of serving it, therefore balancing fetch and delay cost.
This delay pays for fetching~$p$, but not for making room for~$p$.
While evictions are free by themselves, they can cost the algorithm later when an evicted page needs to be re-fetched (\Cref{app:obstacles}, \Cref{ex:room}).
We therefore wait until the delay covers both the cost~\(w_p\) of serving the request and of re-fetching the heaviest page evicted to make space.
To determine the heaviest such page, we observe that the mirror descent path has evicted a page~$i\ne p$ once it reaches length~$Lw_i$:
Using that~$L=\log(1/\delta)$, it holds that
\begin{equation}
    \label{eq:light_evicted_in_band}
    \lpath{x}{p}(s)_i
    =\min\{1,x_ie^{s/w_i}\}
    \geq \min\{1,\delta e^{L}\}
    =1\quad\text{for all } x\in[\delta,1]^n,\ p\ne i,\ s\ge Lw_i.
\end{equation}
Therefore, the critical service to~$p$ displaces all pages of weight up to~$\dur{\lpath{x}{p}}/L$.
To cover the possible re-fetching of these pages in addition to fetching~\(p\), \PMD waits until the delay is at least~$w_p + \dur{\lpath{x}{p}}/L$.

However, if many light pages are simultaneously pending and the cache is filled with heavy pages, following this rule for each page individually means that we let each pending request incur a very large delay cost.
Displacing a heavy page to serve all light pages may be much cheaper than waiting for each light page to accumulate sufficient delay to warrant evicting a heavy page on its own (\Cref{app:obstacles}, \Cref{ex:shared}).
Thus, instead of waiting for the delay of a single request for~\(p\) to cover~$\dur{\lpath{x}{p}}/L$, we consider all pages of weight at most~\(w_p\).
To denote the aggregated delay cost of all pending requests for a page~\(q\) at time~\(\tau\), we define
\[
    \delayof_q(\tau)\coloneqq\sum_{r\text{ pending on }q}d_r(\tau-a_r).
\]
Finally, we need to prevent delay from being charged multiple times, causing more critical service than it can pay for.
Thus, for each page~$q$, we keep track of the amount~$c_q$ of its delay that has been \emph{claimed} to make another page critical and only unclaimed delay contributes to making pages critical.
Whenever~\(q\) is cached,~\(c_q\) resets to~\(0\) because there is no more delay to track.

Putting the above components together, instead of waiting for the delay on~$p$ beyond~\(w_p\) to reach~\(\dur{\lpath{x}{p}}/L\), we wait for the \emph{overshoot}
\[
    o_q\coloneqq\big(\delayof_q(t)-c_q-w_q\big)_+
\]
of all lighter pages to reach this threshold, i.e., a page~$p$ turns \emph{critical} at the last moment before
\begin{equation}\label{eq:trigger}
    o_p > 0\qquad\text{and}\qquad\sum_{q\colon w_q\le w_p}o_q\ge\dur{\lpath{x}{p}}/L.
\end{equation}
Note that there is a well defined last time before \eqref{eq:trigger} is satisfied, because~\(x\) and~\(c_q\) remain constant between two steps of \PMD and the overshoots are continuous, as delays are continuous and new requests initially have~\(0\) delay.

Right as \eqref{eq:trigger} is about to be satisfied, \PMD increases the corresponding values~\(c_q\) by~\(\dur{\lpath{x}{p}}/L\) in total and by at most~$o_q$ each, and performs a critical service to~\(p\).

\paragraph{Preemptive service}
After a critical service, \PMD performs \emph{preemptive services} to pending requests.
A preemptive service has an assigned weight level~$W$ (described below) and serves pending requests of weight at most~$W$ in order of decreasing urgency.
This urgency is measured in the time until the delay of pending requests to a page reaches the weight of serving them.
Formally, for each page~\(q\) we define its \emph{effective deadline} at a time~\(t\) as
\begin{equation}\label{eq:deadline}
    \deadline(q)\coloneqq\inf\big\{\tau\ge t:\delayof_{q}(\tau)>w_{q}\big\},
\end{equation}
so that~\(\delayof_{q}(\tau)>w_{q}\) for all~\(\tau > \deadline(q)\).
\PMD performs preemptive service to pages in order of increasing effective deadlines.

To do so, it first evicts~$1-\delta$ mass of pages of weight at most~$W$ by increasing~\(x_i\) for these pages (\emph{opening a hole}), uses this mass to serve pending requests to pages~$q$ by decreasing~$x_q$ to~$\delta$ and back (\emph{round trips}), and then returns to its previous configuration (\emph{closing the hole}).

Preemptive services take no time and happen at the same time as the critical service that triggered them.
To determine to what extent preemptive service should be performed for level~$W$, \PMD tracks a budget~$B_W$ for each level~$W$ throughout the execution.
\PMD starts a preemptive service for every \emph{weight level}, i.e., every power of two~$W=2^j$, for which the pages of weight at most~$W$ hold a total of at least~$1-\delta$ cached mass, i.e.~$\sum_{i:w_i\leq W}(1-x_i)\geq 1-\delta$ (so that there is sufficient mass to open a hole), and the budget~$B_W$ is at least~$4W$ (i.e., the budget can cover the cost of closing a hole and serving some requests).
It then performs round trips as long as~$B_W\ge W$, so that the budget can cover each round trip.

Each budget~$B_W$ starts with~$B_W=0$ and can increase during critical service.
To see by how much, we split the parameter range of the mirror descent path into disjoint intervals~$I_W \coloneqq [W\cdot L, 2W\cdot L)$.
By \eqref{eq:light_evicted_in_band}, only pages heavier than~$W$ are evicted during~$I_W$ (\Cref{fig:md-coordinates}).
Moreover, we later prove that the rate of weighted eviction is at least~$\frac 13$ throughout (see the proof of \Cref{lem:clock}).
So if~$w_p\le W$, then during~$I_W$ the relatively light page~$p$ causes the eviction of pages heavier than~$W$.
If light pages cause this repeatedly, \PMD may keep evicting a heavy page for light requests, whereas it might have been cheaper to serve the light requests ahead of time and keep the heavy page cached afterwards.
Therefore, we dedicate the length of the mirror descent path in~$I_W$ to preemptive service at level~$W$.

To this end, we write~$\durW{W}{\lambda}\coloneqq|I_W\cap[0,\dur{\lambda}]|$ for the length of the part of~$\lambda$ in~$I_W$, and increase the budget~\(B_W\) of every level~$W\geq w_p$ by~$\durW{W}{\lambda}$ (see \Cref{fig:md-coordinates}).
Since the intervals~$I_W$ are disjoint, a critical service with path~$\lambda$ deposits at most~$\dur{\lambda}$ into all budgets in total.
Note that, since the rate of weighted eviction is at least~$\frac13$ throughout, the deposit~$\send$ is at most three times the weighted eviction (\Cref{lem:clock}).
As the total weighted eviction and fetch cost are equal up to an additive term independent of the input, preemptive service only alters the cost of the algorithm by a constant factor.

\subsection{Feasibility of \PMD}
We establish that \PMD is feasible in the sense that throughout the execution, we never cache more mass than~$k$, budgets stay nonnegative, and every request is eventually served.
\begin{lemma}\label{lem:inv}
    \PMD is feasible in the following sense:
    \begin{enumerate}[label=(\alph*)]
        \item Throughout the execution, it holds that~$x\in[\delta,1]^n$,~$\sum_ix_i\ge n-k$ and~$B_W\ge0$.
              Between steps and after every critical service, it holds that~$\sum_ix_i=n-k$.
        \item Every request is served.
    \end{enumerate}
\end{lemma}
\begin{proof}
    \begin{enumerate}[label=(\alph*)]
        \item By definition of the mirror descent paths~$\lpath{x}{p}$ we have~$x\in[\delta, 1]^n$ and~$\sum_i x_i = n-k$ during critical service.
              When opening a hole during preemptive service, the condition~$\sum_{i:w_i\le W}(1-x_i) \ge 1-\delta$ guarantees that this is possible without violating the bounds~$x\in [\delta, 1]^n$.
              In particular, it gives~$\theta\le1$, so that every~$x_i$ stays at most~$1$.

              During a round trip to page~$i$,~$x_i$ decreases no further than~$\delta$.
              This decreases~$\sum_i x_i$ by at most~$1-\delta$, which is offset by the total increase of~$1-\delta$ from opening the hole.

              When closing the hole \PMD returns to the state from before the preemptive service, which we just showed satisfies the assumptions.

              Decreases to~$B_W$ only occur after conditions that ensure~$B_W$ covers the decrease: closing the hole costs~$W(1-\delta)\le4W$, and a round trip on~$q$ costs~$w_q(x_q-\delta)\le w_q\le W$.

        \item Note first that there is a finite number of times at which a page becomes critical because once a page turns critical, the call to~\FCrit{$p$} performs a critical service to the page, serving a request of which there are only finitely many.
              Furthermore, all loops in the algorithm are over finite sets (assuming we only loop over the relevant levels \(W\in\{1,2,4,\dots,\max_i w_i\}\)).
              Suppose some request~$r$ for a page~$p$ is never served.
              After the last critical service,~$x$ and~$c_p$ remain constant, while~$\delayof_p(t)\ge d_r(t-a_r)\to\infty$ (by unboundedness of the delay cost functions), so~$o_p\to\infty$.
              Therefore, \eqref{eq:trigger} eventually holds and \PMD serves the request.
    \end{enumerate}
\end{proof}

\section{Reducing to fetch cost of critical service}\label{sec:reducing}
The aim of this section is to show that it suffices to analyze the fetch cost~\(\smash{C_{\fetch}^{\crit}=\sum_{\nu\in\Crit}\fetch(\lambda^{(\nu)})}\) of critical services of \PMD, where~$\Crit$ is the set of critical steps and~\(\lambda^{(\nu)}\) is the path of the critical service performed in step~\(\nu\) (\Cref{sec:length}).
More precisely, we define the fractional fetch and delay costs of \PMD, show that their sum corresponds to the expected cost of the randomized algorithm that \PMD is transformed into, and prove that all costs incurred by \PMD are bounded by a constant multiple of
\[
    C_{\fetch}^{\crit} + \sum_i w_i,
\]
i.e., the fetch cost of critical services and an additive term independent of the input (\Cref{prop:cost}).

To do so, we show the following result in \Cref{sec:length}.
\begin{restatable}{lemma}{lengthfetch}\label{lem:lengthfetch}
    The total length of the critical services of \PMD satisfies
    \[
        \sum_{\nu\in\Crit}\dur{\lambda^{(\nu)}}
        \le 3 \cdot C_{\fetch}^{\crit} + 3\sum_i w_i.
    \]
\end{restatable}
Algorithms incur two types of costs: delay costs and fetch costs.
\PMD generates fetch cost during critical and preemptive services.
Since we want to reduce to the fetch costs of critical services, we only need to consider delay costs and the fetch costs of preemptive services.
\PMD's preemptive services are paid for by the budgets~\(B_W\) and therefore incur at most the total increase of~\(B_W\) as fetch costs.
Since every critical service increases the~\(B_W\) by at most its length, the total deposit is at most the length of critical service and we can apply \Cref{lem:lengthfetch}.

Using~\(d_r(\tau_r-a_r)\) as the delay cost of \PMD for a request~\(r\) is not sufficient for our analysis, as we discuss in \Cref{subsec:bounding_fractional_delay_cost}, where we also define the correct notion.

Writing~\(C_\fetch\) and~\(C_\delayof\) for the fractional fetch and delay costs of \PMD, defined below, we formally show
\begin{equation}
    \label{eq:proof_goal_sec_reduction}
    C_\fetch
    \le C_{\fetch}^{\crit} + \sum_{\nu\in\Crit}\dur{\lambda^{(\nu)}}
    \qquad\text{and}\qquad
    C_\delayof
    \le C_\fetch + 5\sum_{\nu\in\Crit} \dur{\lambda^{(\nu)}},
\end{equation}
and \Cref{lem:lengthfetch} converts the lengths into fetch cost of critical services.
We begin by proving \Cref{lem:lengthfetch}.

\subsection{Proof of \texorpdfstring{\Cref{lem:lengthfetch}}{Lemma~\ref*{lem:lengthfetch}}}\label{sec:length}
A fractional algorithm traverses a finite sequence of states~$x^{(0)},x^{(1)},\dots,x^{(N)}\in[0,1]^n$.
For \PMD,~$x^{(0)}$ is the initial state, for a critical service with path~\(\lambda\) there is a state~$x^{(i)} = \lambda(\send)$ and for every preemptive service there is a state after opening the hole, a state when the page of a round trip is cached, a state after returning, and the state after closing the hole.
All these states are in the order in which they occur in \PMD.
We call the move from~$x^{(\nu)}$ to~$x^{(\nu+1)}$ \emph{step~$\nu$}.
If it is a critical service to a page~$p^{(\nu)}$, we call step~$\nu$ \emph{critical} and write~$\lambda^{(\nu)}\coloneqq\lpath{x^{(\nu)}}{p^{(\nu)}}$ for its path.
Let~$\Crit$ be the set of critical steps.
\begin{definition}[fractional fetch cost]\label{def:fracfetch}
    The \emph{fractional fetch cost} of a fractional algorithm traversing the states~$x^{(0)},\dots,x^{(N)}$ is
    \[
        C_{\fetch}\coloneqq\sum_{\nu<N}\sum_iw_i\big(x^{(\nu)}_i-x^{(\nu+1)}_i\big)_+.
    \]
    The fractional fetch cost of the critical steps is~$C_{\fetch}^{\crit}\coloneqq\sum_{\nu\in\Crit}\sum_iw_i\big(x^{(\nu)}_i-x^{(\nu+1)}_i\big)_+$.
\end{definition}
In particular, every critical step~$\nu$ contributes exactly~$\fetch(\lambda^{(\nu)})$, so the fractional fetch cost of critical steps is~$C_{\fetch}^{\crit}=\sum_{\nu\in\Crit}\fetch(\lambda^{(\nu)})$.

We prove \Cref{lem:lengthfetch} in two steps.
First, we compare the length of each critical service with the weighted eviction during it and, second, we compare the total weighted eviction during critical services with their total fetch cost.
We define the weighted eviction during an interval~$J\subseteq[0,\dur{\lambda}]$ of a critical service~\(\lambda\) as
\begin{equation}\label{eq:ev}
    \EV(J)\coloneqq
    \int_J \|\lambda'_+(t)\|_{\ell_1(w)}\dd t,
\end{equation}
and write~$\EV(\lambda)\coloneqq\EV([0,\send])$ for the total weighted eviction along~$\lambda$.
\begin{lemma}\label{lem:clock}
    For every state~$x\in[\delta,1]^n$ with~$\sum_ix_i=n-k$ and every page~$p$,~$\dur{\lpath{x}{p}}\le3\,\EV(\lpath{x}{p})$.
\end{lemma}
\begin{proof}
    Let~$\lambda\coloneqq\lpath{x}{p}$ and let~$E(s)$ be the set of pages fully evicted at~$s$,
    \begin{equation}\label{def:evicted}
        E(s)\coloneqq\{i\neq p:\lambda(s)_i=1\}.
    \end{equation}
    For~$i\notin E(s)\cup\{p\}$,~$\lambda'(s)_i=\lambda(s)_i/w_i$, and all other coordinates are constant or nonincreasing, so~$\|\lambda_+'(s)\|_{\ell_1(w)}=\sum_{i\notin E(s)\cup\{p\}}\lambda(s)_i$.
    By definition of the mirror descent path~$\lambda(s)_i\ge\delta$ for all~\(i\) and~$s\le\dur{\lambda}$, so if there were at least~\(n-k\) pages in~\(E(s)\) then~\(\sum_i\lambda(s)_i \ge n-k + k\delta > n-k\), but the definition of~$\lambda(s)_p$ ensures~$\sum_i\lambda(s)_i=n-k$.
    Thus, there can be at most~$n-k-1$ pages in~$E(s)$ and consequently, there are at least~$k$ pages~$i\notin E(s)\cup\{p\}$.
    Hence, it holds for all~$s\le\dur{\lambda}$ that
    \[
        \|\lambda_+'(s)\|_{\ell_1(w)}
        =\sum_{i\notin E(s)\cup\{p\}}\lambda(s)_i
        \ge\sum_{i\notin E(s)\cup\{p\}}\delta
        \ge k\delta
        =\frac{k}{2k+1}\ge\frac13.
    \]
    Integrating over~$[0,\dur{\lambda}]$ gives~$\EV(\lambda)\ge\dur{\lambda}/3$.
\end{proof}
We argue that the total weighted eviction during critical services differs from their total fetch cost by at most an additive term independent of the input.
\begin{lemma}\label{lem:evfetch}
    Over the entire execution, the total weighted eviction during critical services satisfies
    \[
        \sum_{\nu\in\Crit}\EV(\lambda^{(\nu)})
        \le C_{\fetch}^{\crit}+\sum_iw_i.
    \]
\end{lemma}
\begin{proof}
    Let~$\lambda$ be the path of some critical service to~$p$.
    By definition,~$\lambda(s)_i$ is nondecreasing for~$i \ne p$ and only~$\lambda(s)_p$ decreases, so the integral defining~$\EV$ evaluates to
    \[
        \EV(\lambda)
        =\sum_{i\ne p}w_i\big(\lambda(\send)_i-\lambda(0)_i\big)
        \qquad\text{and}\qquad
        \EV(\lambda)-\fetch(\lambda)
        =\sum_iw_i\big(\lambda(\send)_i-\lambda(0)_i\big).
    \]
    Summed over all critical services, the right-hand side telescopes to at most~$\sum_i w_i$, since every preemptive service ends in the state from before it (proof of \Cref{lem:inv}).
    Thus, over the entire execution, the total weighted eviction is~$\sum_{\nu\in\Crit}\EV(\lambda^{(\nu)})\le C_{\fetch}^{\crit}+\sum_iw_i$.
\end{proof}
Together, the two lemmas yield \Cref{lem:lengthfetch}.
\lengthfetch*
\begin{proof}
    By \Cref{lem:clock,lem:evfetch},
    \[
        \sum_{\nu\in\Crit}\dur{\lambda^{(\nu)}}
        \le 3\sum_{\nu\in\Crit}\EV(\lambda^{(\nu)})
        \le 3\cdot C_{\fetch}^{\crit} + 3\sum_i w_i.
        \qedhere
    \]
\end{proof}

\subsection{Bounding Fetch Cost}\label{subsec:bounding_fetch_cost}
In this section we show the first inequality of \eqref{eq:proof_goal_sec_reduction}, i.e., we bound the fractional fetch cost of \PMD in terms of the length of critical services.
The critical steps contribute~$C_{\fetch}^{\crit}$ to~$C_{\fetch}$ and the remaining fetch cost is incurred by preemptive services.
\begin{proposition}\label{prop:preemptfetch}
    The fractional fetch cost of \PMD satisfies
    \[
        C_{\fetch}
        \le C_{\fetch}^{\crit} + \sum_{\nu\in\Crit}\dur{\lambda^{(\nu)}}.
    \]
\end{proposition}
\begin{proof}
    As discussed above, \PMD incurs fetch cost during critical and preemptive services.
    Thus, it suffices to bound the latter by the length~\(\sum_{\nu\in\Crit}\dur{\lambda^{(\nu)}}\) of critical services.

    Opening the hole and returning from a round trip only increase coordinates of~$x$, so a preemptive service at level~$W$ incurs fetch cost only during round trips and to restore the original state, both of which are paid for by a decrease in~$B_W$:
    Round trips on pages~\(q\) with~\(w_q \le W\) are only performed as long as~\(B_W \ge W\) and paid for immediately, while closing the hole costs at most~\(W(1-\delta)\), which was paid at the start of the preemptive service.
    Thus, it suffices to bound the total decrease of all budgets, or, because~\(B_W\ge0\), the total increase.

    A critical service for a page~$p$, with path~$\lambda$, increases the budget~$B_W$ by~$\durW{W}{\lambda}$ for all~$W\ge w_p$.
    Since the intervals~$I_W$ are disjoint, the total increase is at most~$\dur{\lambda}$.
    Over the entire execution, the increase is hence bounded by~$\sum_{\nu\in\Crit}\dur{\lambda^{(\nu)}}$, which bounds the contribution of preemptive services to~$C_{\fetch}$, and the claim follows.
\end{proof}

\subsection{Bounding Fractional Delay Cost}\label{subsec:bounding_fractional_delay_cost}
Having bounded~\(C_{\fetch}\), it remains to show the second inequality in \eqref{eq:proof_goal_sec_reduction}, i.e.,
\[
    C_\delayof
    \le C_\fetch+5\sum_{\nu\in\Crit}\dur{\lambda^{(\nu)}}.
\]
To define the fractional delay cost~\(C_{\delayof}\), consider some request~\(r\) and let~\(\tau_r\) be the time at which~\(r\) is served, i.e., the first time since the release for which~\(x_{p_r} = \delta\).
We could naively define the delay cost to be~\(d_r(\tau_r-a_r)\), but that would ignore the fact~\(p_r\) may already be cached with some probability when~\(r\) is released.
In this case, some realizations of the randomized algorithm would not incur any delay cost at all thus lowering the expected cost of the algorithm.
Hence, the naive definition does not represent the expected delay cost of the randomized algorithm.
Instead, we weigh the delay by the factor~$\alpha_r$ representing the antimass of~$p_r$ when~$r$ arrives, i.e., the probability of the request not being served upon arrival.
\begin{definition}[fractional delay cost]\label{def:fracdelay}
    For a request~$r$, let~$x$ be the state of \PMD right before~\(r\) arrives at time~$a_r$.
    Let~$\alpha_r\coloneqq x_{p_r}-\delta$ be the missing mass on~\(p_r\) at that time and let~$\tau_r$ be the time at which~$r$ is served.
    The \emph{fractional delay cost} of~$r$ is~\(\alpha_r\cdot d_r(\tau_r-a_r)\) and the \emph{total fractional delay cost} of \PMD is defined as
    \[
        C_{\delayof}\coloneqq\sum_r\alpha_rd_r(\tau_r-a_r).
    \]
\end{definition}
We will show later that this definition captures the delay cost of the randomized algorithm corresponding to \PMD up to a factor of~$3$ (\Cref{prop:rounding}).
We want to consider the delay not as a sum over requests but as a sum over steps.
To this end, we define a \emph{visit to~\(p\)} as a step in which a request to \(p\) is served, write~$\Visit$ for the set of visits, and make the following observation, which is immediate by the algorithm's definition.
\begin{observation}\label{obs:monotone}
    Between two visits of a page~$p$ its antimass~$x_p$ never decreases (outside of the temporary movements during preemptive service).
    In particular, every request~$r$ on~$p$ satisfies~$x_p-\delta\ge\alpha_r$ until it is served.
\end{observation}
Right before \PMD visits a page~\(p\) in a step~\(\nu\), the overshoot is~\(o_p = \big(\delayof_p(t)-c_p-w_p\big)_+\), so the delay of all pending requests to~\(p\) is bounded by
\[
    \sum_{r\in R_\nu} d_r(t - a_r)
    = \delayof_p(t)
    \le w_p + c_p + o_p,
\]
where~\(R_\nu\) is the set of all requests served in step~\(\nu\).
In particular, we can bound the fractional delay cost of all~\(r\in R_\nu\) as
\[
    \sum_{r\in R_\nu} \alpha_r \cdot d_r(\tau_r - a_r)
    \stackrel{\text{Obs.}~\ref{obs:monotone}}{\le}
    \sum_{r\in R_\nu} (x_p - \delta) \cdot d_r(\tau_r - a_r)
    \le (x_p-\delta)\cdot(w_p + c_p + o_p)
    \stackrel{x_p\le1}{\le} (x_p-\delta)\cdot w_p + c_p + o_p,
\]
i.e., the delay cost of all requests served with a visit is at most the cost of fetching the page, its claimed delay and its overshoot.
Since all requests are served in some step, we can sum this over all visits~\(\nu\) to a page~\(p^{(\nu)}\) to get
\begin{align*}
    C_{\delayof}
    = \sum_{\nu\in\Visit}\sum_{r\in R_\nu} \alpha_r \cdot d_r(\tau_r - a_r)
     & \le\sum_{\nu\in\Visit} \left((x_{p^{(\nu)}}-\delta)\cdot w_{p^{(\nu)}} + c_{p^{(\nu)}} + o_{p^{(\nu)}}\right) \\
     & \le C_{\fetch} + \sum_{\nu\in\Visit} (c_{p^{(\nu)}} + o_{p^{(\nu)}}).
\end{align*}
With this, it suffices to bound the last sum in terms of~\(\sum_{\nu\in\Crit} \dur{\lambda^{(\nu)}}\).
This is easy for the summand~\(c_p\).
Because~\(c_p\) is set to~\(0\) after every visit, any increase in~\(c_p\) is counted at most once and the claims only increase by~\(\dur{\lambda}/L\) in total for a critical service with path~\(\lambda\).
Thus, the entire sum is bounded by~\(\sum_{\nu\in\Visit} c_{p^{(\nu)}} \le \sum_{\nu\in\Crit}\dur{\lambda^{(\nu)}}/L\).

Bounding~\(\sum_{\nu} o_p\) is more involved.
We charge the overshoot paid by all visits at a time~$t$ to the lengths of the critical services that \PMD performs at the same time~$t$.

Let~$x$ be the state right before the steps of \PMD at time~$t$.
At time~\(t\) for every page~$p$ with~\(o_p > 0\), we have~$\sum_{q: w_q\le w_p} o_q \le \frac1L\dur{\lpath{x}{p}}$ because otherwise a critical service would have been triggered earlier.
If \PMD performed only a single critical service to a page~\(p\) at time~$t$ then~\(o_p\) would be bounded by~\(\frac1L\dur{\lpath{x}{p}}\) and we would be done.

The difficulty is that \PMD may perform many services at time~$t$.
After the first critical service, on a page~$p$, the state jumps from~$x$ to some state~$\tilde x$.
The bound on the overshoot of page~$q$ still refers to the critical service~$\lpath{x}{q}$, whereas the service that \PMD actually performs on~$q$ starts in~$\tilde x$, and~$\dur{\lpath{\tilde x}{q}}$ might be smaller than~$\dur{\lpath{x}{q}}$.
However, the jump from~$x$ to~$\tilde x$ is caused by a critical service to~$p$.
It therefore suffices to relate the length change between~$\dur{\lpath{x}{q}}$ and~$\dur{\lpath{\tilde x}{q}}$ to the length of the critical service to~$p$.
To do so, we start by showing an analog to the triangle inequality for lengths of critical services.
\begin{lemma}\label{lem:drop}
    Let~\(p\) be a page and let~\(x\) and~\(\tilde x\) be the states before and after a critical service to~\(p\) is performed.
    Let~\(q\ne p\) be another page with~\(x_q>\delta\).
    Then performing a critical service to~\(q\) from~\(x\) is no more expensive than performing the service to~\(p\) and then the service to~\(q\) from~\(\tilde x\), i.e.,~\(\dur{\lpath{x}{q}}\le\dur{\lpath{x}{p}} + \dur{\lpath{\tilde x}{q}}\).
\end{lemma}
\begin{proof}
    Let~$\lambda\coloneqq\lpath{x}{q}$ and~$\tilde\lambda\coloneqq\lpath{\tilde x}{q}$, so~$\lambda(s)_i=\min\{1,x_ie^{s/w_i}\}$ and~$\tilde\lambda(s)_i=\min\{1,\tilde x_ie^{s/w_i}\}$ for~$i\ne q$.

    For a page~$i\notin\{p,q\}$, the critical service~\(\lpath{\tilde x}{q}\) continues the mirror descent path that~\(\lpath{x}{p}\) took: if~$i$'s antimass reached~$1$, it stays constant.
    Otherwise,~$\tilde x_i=x_i e^{\dur{\lpath{x}{p}}/w_i}$, so~$\tilde\lambda(s)_i=x_i e^{(s+\dur{\lpath{x}{p}})/w_i}$.
    In particular,~$\tilde\lambda(s)_i= \lambda(s+\dur{\lpath{x}{p}})_i$ for all~$s\ge0$.

    For~\(p\), the two mirror descent paths have different starting points~$\tilde x_p\le x_p$, and since~$p$'s weight does not change, the paths always satisfy~$\tilde\lambda(s)_p\le\lambda(s)_p\le\lambda(s+\dur{\lpath{x}{p}})_p$.
    At~$\dur{\lpath{\tilde x}{q}}$, we have
    \[
        n-k-\delta
        =\sum_{i\ne q}\tilde\lambda(\dur{\lpath{\tilde x}{q}})_i
        \le\sum_{i\ne q}\lambda(\dur{\lpath{\tilde x}{q}}+\dur{\lpath{x}{p}})_i=n-k-\lambda(\dur{\lpath{\tilde x}{q}}+\dur{\lpath{x}{p}})_q,
    \]
    so~$\lambda(\dur{\lpath{\tilde x}{q}}+\dur{\lpath{x}{p}})_q\le\delta$.
    As~$s\mapsto\lambda(s)_q$ is continuous and nonincreasing, and~$\dur{\lpath{x}{q}}$ is defined as the first~$s$ with~$\lambda(s)_q=\delta$, it follows that~$\dur{\lpath{x}{q}}\le\dur{\lpath{\tilde x}{q}}+\dur{\lpath{x}{p}}$.
\end{proof}
Using the above lemma, we can now bound the fractional delay cost of \PMD in terms of the fetch cost and the length of critical services.
\begin{proposition}\label{prop:algdelay}
    The fractional delay cost of \PMD satisfies~$C_{\delayof} \le C_\fetch + 5\sum_{\nu\in\Crit}\dur{\lambda^{(\nu)}}$.
\end{proposition}
\begin{proof}
    We follow the plan outlined above.
    Every request with~$\alpha_r>0$ is served by a visit of its page.
    Since no requests arrive during the steps of \PMD at a time~$t$, every page is visited at most once at time~\(t\).
    Consider a visit of~$p$ at time~$t$, and let~$x$ and~$c_p$ be the state and the claimed delay of~$p$ right before it.
    By \Cref{obs:monotone}, the served requests have~$\alpha_r\le x_p-\delta\le1$, and their delays sum to~$\delayof_p(t)\le w_p+c_p+o_p$, where~$o_p$ is the overshoot right before the steps at time~$t$.
    So the delay of all requests served by the visit is at most~$(x_p-\delta)\cdot\delayof_p(t) \le w_p(x_p-\delta)+c_p+o_p$ and therefore summing over all visits we get
    \[
        C_{\delayof}
        \le C_{\fetch} + \sum_{\nu\in\Visit} (c_{p^{(\nu)}} + o_{p^{(\nu)}})
        \le C_{\fetch} + \sum_{\nu\in\Crit}\dur{\lambda^{(\nu)}} + \sum_{\nu\in\Visit} o_{p^{(\nu)}}.
    \]
    The claim now follows from
    \[
        \sum_{\nu\in\Visit} o_{p^{(\nu)}}
        \le 4\sum_{\nu\in\Crit}\dur{\lambda^{(\nu)}},
    \]
    which we show by grouping the steps by the time at which they happen.
    As above, let~$x$ be the state right before the steps at some time~$t$ and let~$q$ be the heaviest page with~$o_q>0$ that is visited at time~$t$ (if no such page exists the sum of overshoots is zero and there is nothing to show).
    As \eqref{eq:trigger} was not satisfied for~$q$ before time~$t$ we have
    \[
        \sum_{q'\colon w_{q'}\le w_q}o_{q'}
        \le \frac1L\dur{\lpath{x}{q}}.
    \]
    Note that this sum contains the overshoots of all visits at time~\(t\) because~\(q\) is the heaviest page, so we aim to bound it by~\(\sum_{\nu\in\CritAt{t}}\dur{\lambda^{(\nu)}}\), where~$\CritAt{t}\subseteq\Crit$ is the set of critical steps at time~$t$.
    Let~$\tilde x$ be the state in which \PMD starts the service that visits~$q$, and let~$\nu_1,\dots,\nu_i$ be the critical steps at time~$t$ before this service.
    Since preemptive services end in the state in which they start,~$\tilde x$ is obtained from~$x$ by following~$\nu_1,\dots,\nu_i$.
    Applying \Cref{lem:drop}~$i$ times we have
    \[
        \frac1L\dur{\lpath{x}{q}}
        \le \frac1L\left(\sum_{j=1}^{i} \dur{\lambda^{(\nu_j)}}
        + \dur{\lpath{\tilde x}{q}}\right)
        \le \frac1L\left(\sum_{\nu\in\CritAt{t}}\dur{\lambda^{(\nu)}}
        + \dur{\lpath{\tilde x}{q}}\right).
    \]
    If the visit to~$q$ is a critical service, then we have
    \begin{equation}
        \label{eq:prop11_critical_service}
        \sum_{q'\colon w_{q'}\le w_q}o_{q'}
        \le \frac1L\left(\sum_{\nu\in\CritAt{t}}\dur{\lambda^{(\nu)}}
        + \dur{\lpath{\tilde x}{q}}\right)
        \le \frac2L\sum_{\nu\in\CritAt{t}}\dur{\lambda^{(\nu)}}
        \le 2\sum_{\nu\in\CritAt{t}}\dur{\lambda^{(\nu)}}.
    \end{equation}
    If~\(q\) is visited by a round trip of a level-\(W\) preemptive service, we show~\(\dur{\lpath{\tilde x}{q}}\le LW\).
    In other words, if we served~\(q\) with a critical service instead of a preemptive service, the path would have length at most~\(LW\).
    This holds precisely because we ensure light pages can be evicted to make space.
    Formally, preemptive service only happens if~\(\sum_{i:w_i \le W}(1-\tilde x_i) \ge 1-\delta\).
    Subtracting~\(1 - \tilde x_q\) we get that the remaining pages hold at least weight~\(\sum_{i\ne q:w_i \le W}(1-\tilde x_i) \ge \tilde x_q-\delta\).
    These pages are evicted along the mirror descent path and by \eqref{eq:light_evicted_in_band}, all of them are fully evicted at~\(\lpath{\tilde x}{q}(LW)\).
    Since~\(q\) is fetched by the amount that the other pages are evicted, it is fully loaded and the path ends, i.e.,~\(\dur{\lpath{\tilde x}{q}}\le LW\).
    Plugging this in gives
    \begin{equation}
        \label{eq:prop11_preemptive_service}
        \sum_{q'\colon w_{q'}\le w_q}o_{q'}
        \le \frac1L\left(\sum_{\nu\in\CritAt{t}}\dur{\lambda^{(\nu)}}
        + \dur{\lpath{\tilde x}{q}}\right)
        \le \sum_{\nu\in\CritAt{t}}\dur{\lambda^{(\nu)}} + W
        \le \sum_{\nu\in\CritAt{t}}\dur{\lambda^{(\nu)}} + 2W(1-\delta),
    \end{equation}
    where we used~$L\ge1$ and~$\delta\le\frac12$.
    Note that~\(W(1-\delta)\) is the payment of this preemptive service for closing its hole.
    If we now sum \eqref{eq:prop11_critical_service} or \eqref{eq:prop11_preemptive_service} for every time~\(t\), depending on how the heaviest page is served, we get
    \[
        \sum_{\nu\in\Visit}o_{p^{(\nu)}}
        \le 2\sum_{\nu\in\Crit}\dur{\lambda^{(\nu)}} + 2H,
    \]
    where~\(H\) is the total payment of all preemptive services for closing the hole.
    Since~\(H\) is paid from the budgets~\(B_W\), it is bounded by their total increase, which (as we already argued in the proof of \Cref{prop:preemptfetch}) is bounded by~\(\sum_{\nu\in\Crit}\dur{\lambda^{(\nu)}}\).

    Putting everything together we find
    \[
        C_{\delayof}
        \le C_{\fetch} + \sum_{\nu\in\Crit}\dur{\lambda^{(\nu)}} + \sum_{\nu\in\Visit} o_{p^{(\nu)}}
        \le C_{\fetch} + 3\sum_{\nu\in\Crit}\dur{\lambda^{(\nu)}} + 2H
        \le C_{\fetch} + 5\sum_{\nu\in\Crit}\dur{\lambda^{(\nu)}}.\qedhere
    \]
\end{proof}

\subsection{Combined bound}
Finally, we combine the bounds on the fetch and delay cost in \eqref{eq:proof_goal_sec_reduction} with \Cref{lem:lengthfetch} to reduce the analysis of \PMD to bounding the fetch cost of critical services.
\begin{proposition}\label{prop:cost}
    If for every problem instance~$C_{\fetch}^{\crit}\le\rho\cdot\OPT+c$ for some~\(\rho, c\ge 0\), then the fractional fetch cost~$C_{\fetch}$ and the total fractional delay cost~$C_{\delayof}$ of \PMD satisfy
    \[
        C_{\fetch}+C_{\delayof}\le23\rho\cdot\OPT+c',\qquad c'\coloneqq23c+21\sum_iw_i.
    \]
\end{proposition}
\begin{proof}
    We combine our results to get
    \begin{align*}
        C_{\fetch}+C_{\delayof}
         & \stackrel{\text{Prop~\ref{prop:algdelay}}}{\le}2\cdot C_{\fetch} + 5\sum_{\nu\in\Crit}\dur{\lambda^{(\nu)}}              \\
         & \stackrel{\text{Prop~\ref{prop:preemptfetch}}}{\le} 2\cdot C_{\fetch}^{\crit} + 7\sum_{\nu\in\Crit}\dur{\lambda^{(\nu)}}
        \stackrel{\text{Lem~\ref{lem:lengthfetch}}}{\le} 23\cdot C_{\fetch}^{\crit} + 21\sum_i w_i,
    \end{align*}
    and the claim follows by the assumption~$C_{\fetch}^{\crit}\le\rho\cdot\OPT+c$.
\end{proof}

\section{The Bregman potential}\label{sec:bregman}

As shown in the previous section (\Cref{prop:cost}), it suffices to bound~\(C_{\fetch}^{\crit}\) in terms of \OPT.
To do so, we use a potential function argument based on the \emph{Bregman divergence}.
For a state~$x$ let
\[
    D(\OPT\|x)\coloneqq\sum_{i\notin \OPT}w_i\,g(x_i)+\sum_{i\in \OPT}w_ix_i,
    \qquad g(u) \coloneqq u-1-\log u.
\]
This is the Bregman divergence of the weighted entropy function~$\sum_iw_ix_i\log x_i$ between \OPT's state and~$x$ and can be used as a measure for the difference between \PMD's and \OPT's cache.
We use the Bregman divergence as a potential for analyzing \PMD and denote it~$D$ for short when the state~$x$ is clear from context.
For a critical service with path~$\lambda$, we let~$\Delta D(\lambda)$ denote the change of~$D$ during this service.

To bound the total cost of \PMD against \OPT's cost, note first that we can equivalently charge \OPT's evictions instead of its fetches.
Here, we write $\OPT_0$ for the initial cache of \OPT.
\begin{observation}\label{obs:evictions}
    The total cost of \OPT's evictions is at most the fetch cost of \OPT plus~$w(\OPT_0)$, since each eviction is either of a previously fetched page or of a page of the initial cache.
\end{observation}

To begin the potential function argument, we first observe that the potential~$D$ increases by at most~$\O(L)$ per unit of cost that \OPT incurs.
\begin{lemma}\label{lem:optmoves}
    The total increase of~$D$ caused by \OPT's moves is at most~$\O(L)\cdot\OPT+\O(L\sum_i w_i)$.
\end{lemma}
\begin{proof}
    When \OPT evicts~$i$,~$D$ changes by
    $$
        w_i(g(x_i)-x_i)
        = w_i(-1-\log x_i)
        \le w_i\log\frac1\delta
        = L \cdot w_i.
    $$
    When it fetches~$i$,~$D$ changes by~$w_i(x_i-g(x_i))\le w_i$, so the claim follows by \Cref{obs:evictions}.
\end{proof}

Next, we analyze the potential change on a move of \PMD.
Recall that we only consider critical services.
For this, we make a case distinction based on whether the critical page is in \OPT's cache at the time of the service (so that it contributes only to the second sum of~$D$), or not (so that it contributes to the first sum).
\begin{definition}[hit and miss services]
    A critical service to page~$p$ is a \emph{hit service} if~$p\in \OPT$ and a \emph{miss service} if~$p\notin \OPT$.
    We use these terms also for its step, and write~$\Hit$ and~$\Miss$ for the sets of hit and miss steps, so that~$\Crit=\Miss \sqcup \Hit$.
    Their fractional fetch costs~$C_{\fetch}^{\hit}\coloneqq\sum_{\nu\in\Hit}\fetch(\lambda^{(\nu)})$ and~$C_{\fetch}^{\miss}\coloneqq\sum_{\nu\in\Miss}\fetch(\lambda^{(\nu)})$ then satisfy~$C_{\fetch}^{\crit}=C_{\fetch}^{\hit}+C_{\fetch}^{\miss}$.
\end{definition}

For miss services, it turns out that the potential can increase during the service, so that we have to bound this additional increase in our potential function argument.
We give a name to this excess term that needs to be bounded independently:

\begin{definition}[excess]
    For a miss service on page~$p$ in step~$\nu$, with~$w_{\min}(\OPT)$ the lightest weight in~$\OPT$, the \emph{excess} of this service is
    \[
        \Exc(\nu)\coloneqq
        \sum_{\substack{\mathclap{w_p\le W<w_{\min}(\OPT)}}}\durW{W}{\lambda^{(\nu)}}.
    \]
    The \emph{total excess} is~$\Exc\coloneqq\sum_{\nu\in\Miss}\Exc(\nu)$.
\end{definition}

This term corresponds to the excess term in \eqref{eq:bregman_guarantee} as~$1$ integrated over the path length.
However, we only account for the parts of the path that are in intervals~$I_W$ with~$w_p\leq W<w_{\min}(\OPT)$, and bound the remainder separately.

Finally, we compute the potential change during critical services.
The result says that the cost of a hit service is covered by a decrease in potential so that we only need to consider miss services.
\begin{proposition}\label{prop:phaseD}
    Consider a critical service with path~$\lambda$ and triggering page~$p$, performed in step~$\nu$.
    If the service is a hit service,~$\Delta D(\lambda)\le-\fetch(\lambda)$.
    If it is a miss service,~$\Delta D(\lambda)\le 2L\cdot w_p+\Exc(\nu)$.
\end{proposition}
\begin{proof}
    Recall that we assume that \OPT's cache is constant during the service and that it contains exactly~$k$ pages.
    Moreover, recall that~$\fetch(\lambda) = w_p(x_p - \delta)$.

    The missing mass of all pages~$i\ne p$ is increased, which increases the second sum in~$D$ and decreases the first.
    Consider the rate of this change at a point~$s$ of the mirror descent path.
    First, observe that pages~$i\ne p$ with~$\lambda(s)_i=1$ are already fully evicted, so their missing mass~$\lambda(s)_i$ has a constant contribution to~$D$.
    Define the set of evicted pages~$E(s)$ as in \eqref{def:evicted}.
    The missing mass of a page~${i\notin E(s)\cup \{p\}}$ with~$\lambda(s)_i \ne 1$ increases at rate~$\lambda(s)_i/w_i$ (recall that we set~$\lambda(s)_i = x_ie^{s/w_i}$), which changes~$D$ at rate~$w_i\cdot \lambda(s)_i/w_i=\lambda(s)_i$ if~$i\in \OPT$ and at rate
    \[
        w_i\cdot g'(\lambda(s)_i)\cdot \frac{\lambda(s)_i}{w_i}
        = w_i\cdot \left(1-\frac{1}{\lambda(s)_i}\right)\cdot \frac{\lambda(s)_i}{w_i}= \lambda(s)_i - 1
    \]
    if~$i\notin \OPT$.
    Together, the pages~$i\ne p$ change~$D$ at rate
    \begin{align*}
         & \sum_{i\notin \OPT \cup \{p\}}(\lambda(s)_i-1)+\sum_{i \in \OPT \setminus(E(s)\cup\{p\})}\lambda(s)_i                                                     \\
         & \quad = \sum_{i\notin \OPT \cup \{p\}} \lambda(s)_i - (n-k - \mathds{1}_{p\notin \OPT}) + \sum_{i \in \OPT \setminus\{p\}}\lambda(s)_i - |E(s) \cap \OPT| \\
         & \quad = \sum_{i\ne p} \lambda(s)_i - (n-k - \mathds{1}_{p\notin \OPT}) - |E(s) \cap \OPT|                                                                 \\
         & \quad = \mathds{1}_{p\notin \OPT}-\lambda(s)_p-|E(s)\cap \OPT|.
    \end{align*}
    The function~$s\mapsto D(\OPT\|\lambda(s))$ is absolutely continuous, so we may integrate these rates to get the total change in~$D$.

    If the service is a hit service, then~$p\in \OPT$ contributes~$w_px_p$ to the second sum in~$D$, which decreases by~$\fetch(\lambda)$ during the service.
    The remaining pages change~$D$ at rate~$-\lambda(s)_p-|E(s)\cap \OPT|\le0$, and the claim follows.

    If it is a miss service, then~$p\notin \OPT$ contributes~$w_pg(x_p)$ to the first sum, hence changing~$D$ by
    \[
        w_p\big(g(\delta)-g(x_p)\big)\le w_pg(\delta)<Lw_p
    \]
    in total.
    It remains to bound the change caused by the pages~$i\ne p$ by~$Lw_p+\Exc(\nu)$.
    As we showed above, the remaining pages change~$D$ at rate
    \[
        h(s)\coloneqq 1-\lambda(s)_p-|E(s)\cap \OPT|\le 1.
    \]
    For~$s\ge L\cdot w_{\min}(\OPT)$, \eqref{eq:light_evicted_in_band} shows that the lightest page in \OPT is evicted, i.e., it lies in~$E(s)\cap\OPT$.
    Therefore,~$h(s)\le-\lambda(s)_p<0$.
    Hence the total change caused by the pages~$i\ne p$ is
    \[
        \int_0^{\dur{\lambda}}h(s)\,\dd s\le\big|[0,\dur{\lambda}]\cap[0,L\cdot w_{\min}(\OPT))\big|.
    \]
    Since~$[0,L\cdot w_{\min}(\OPT))\subseteq[0,Lw_p)\cup[Lw_p,L\cdot w_{\min}(\OPT))$, the right-hand side is at most
    \[
        Lw_p
        +\big|[0,\dur{\lambda}]\cap[Lw_p,L\cdot w_{\min}(\OPT))\big|
        =Lw_p
        +\sum_{w_p\le W<w_{\min}(\OPT)}\durW{W}{\lambda}
        = Lw_p+\Exc(\nu).
    \]
    This proves the claim.
\end{proof}

\section{Charging to \texorpdfstring{\OPT}{OPT}}\label{sec:charging}

While the potential decrease pays for the cost of hit services, the potential can increase during miss services.
In particular, the amortized fetch cost of a miss service on page~$p$ is~$\O(L)\cdot w_p+\Exc(\nu)$.
In this section, we argue that the first term can be charged to \OPT, and give a charging argument that will later be reused to bound the excess.

By definition of a miss service, \OPT does not cache the service's page~$p$ at the time of the service.
Moreover, since the service was triggered at all, the delay of the requests to~$p$ has reached at least its weight~$w_p$.
If \OPT has already served a request to~$p$ and evicted the page since, we can charge~$w_p$ to \OPT's eviction.
Otherwise, \OPT serves the requests to~$p$ later and pays at least their current delay~$w_p$.
Either way, we can charge the quantity~$w_p$ to \OPT.

To make this precise, we define a notion of chargeable visits.
This extends the argument from miss services also to preemptive round trips.

\begin{definition}[chargeable visits]\label{def:chargeable}
    A visit of a page~$p$ is \emph{chargeable} if, right before it,
    \begin{enumerate}[label=(\roman*)]
        \item\label{item:abandoned} $p\notin \OPT$ and \OPT has served some request to~$p$ that is pending for \PMD, or
        \item\label{item:optdelay} the total delay cost of \OPT for the requests to~$p$ that are pending for \PMD is at least~$w_p$.
    \end{enumerate}
\end{definition}

For visits satisfying \ref{item:optdelay}, it is easy to see that their weight can be charged to \OPT's delay cost.
For visits satisfying \ref{item:abandoned}, \OPT must have evicted the visited page after serving it.
As a consequence, we can charge the weight of the visit to \OPT as follows.
\begin{lemma}\label{lem:chargeable}
    The sum of the weights~$w_p$ of all chargeable visits on pages~$p$ is at most~$\OPT+w(\OPT_0)$.
\end{lemma}
\begin{proof}
    Let~$\mu$ be a visit on~$p$ at time~$t$ satisfying \ref{item:abandoned}.
    Charge the cost~$w_p$ of~$\mu$ to \OPT's last eviction of~$p$ before~$\mu$.

    Suppose two visits~$\mu_1$ before~$\mu_2$ on~$p$ satisfying \ref{item:abandoned} are charged to the same eviction~$e$.
    Then \OPT did not evict~$p$ between~$\mu_1$ and~$\mu_2$.
    Since~$p\notin \OPT$ upon~$\mu_2$, \OPT does not have~$p$ in its cache between~$\mu_1$ and~$\mu_2$.
    Note that every request pending for \PMD at~$\mu_2$ must arrive after~$\mu_1$.
    In particular, no request pending for \PMD at~$\mu_2$ has been served by \OPT at~$\mu_2$, contradicting \ref{item:abandoned} for~$\mu_2$.
    Therefore, each charged quantity~$w_p$ corresponds to a distinct eviction of \OPT and the total weight of visits satisfying \ref{item:abandoned} is bounded by the total weight of \OPT's evictions.
    By \Cref{obs:evictions}, this is at most the fetch cost of \OPT plus~$w(\OPT_0)$.

    Every other chargeable visit satisfies \ref{item:optdelay}, i.e., it serves requests for which \OPT pays at least~$w_p$ of delay, and distinct visits serve disjoint sets of requests, so these visits have total weight at most the delay cost of \OPT.
\end{proof}

We next argue that miss services are chargeable, so that we can bound the term~$\O(L)\cdot w_p$ of their amortized cost by \OPT's cost.

\begin{lemma}\label{lem:miss}
    $\sum_{\nu\in\Miss}w_{p^{(\nu)}}\le\OPT+w(\OPT_0)$.
\end{lemma}
\begin{proof}
    We show that every miss service on a page~$p$ at time~$t$ is a chargeable visit, the claim then follows from \Cref{lem:chargeable}.
    A page~\(p\) becomes critical at the last moment before~\eqref{eq:trigger} is satisfied, so by continuity we have~$\delayof_p(t)\ge w_p$.
    If \OPT has served some pending request to~$p$, then the service satisfies \ref{item:abandoned}, since~$p\notin\OPT$.
    Otherwise, \OPT pays at least~$\delayof_p(t)\ge w_p$ in delay cost, so the service satisfies \ref{item:optdelay}.
\end{proof}

We can combine the above lemmas to bound the total fetch cost through a potential function argument.
The fetch cost of hit services is covered by the potential decrease (\Cref{prop:phaseD}).
The potential increase caused by \OPT's moves is at most~$\O(L)$ times \OPT's cost (\Cref{lem:optmoves}).
The fetch cost of miss services, as well as the term~$\O(L)\cdot w_p$ from the potential increase can be charged to \OPT directly (\Cref{lem:miss}).
In order to bound the total fetch cost only in terms of \OPT, it then remains only to bound the total excess.

\begin{proposition}\label{prop:tele}
    $C_{\fetch}^{\crit} = \O(L)\cdot \OPT+\Exc+\O(L\sum_i w_i)$.
\end{proposition}
\begin{proof}
    The potential~$D$ changes only in critical services and at \OPT's moves, since preemptive services return to the same~$x$ with \OPT's cache constant.
    By \Cref{lem:optmoves,prop:phaseD} and since~$D\geq 0$ throughout, the total decrease in potential is at least~$C_{\fetch}^{\hit}$, and the total increase is at most
    \[
        2L\cdot \sum_{\nu\in\Miss}w_{p^{(\nu)}}+\Exc+\O(L)\cdot \OPT+\O(L\sum_i w_i).
    \]
    Clearly, the total decrease is at most the total increase plus the initial potential.
    The initial potential is~$D\le(L+1)\sum_iw_i=\O(L\sum_i w_i)$ since~$g(x_i)<L$ and~$x_i\le1$.
    It follows that
    \[
        C_{\fetch}^{\hit} \leq 2L\cdot \sum_{\nu\in\Miss}w_{p^{(\nu)}}+\O(L)\cdot \OPT+\Exc+\O(L\sum_i w_i).
    \]
    Applying \Cref{lem:miss} to the above inequality and to~$C_{\fetch}^{\miss}\leq \sum_{\nu\in\Miss}w_{p^{(\nu)}}$ gives
    \[
        C_{\fetch}^{\crit} = C_{\fetch}^{\hit}+C_{\fetch}^{\miss} \leq \O(L)\cdot \OPT+\Exc+\O(L\sum_i w_i). \qedhere
    \]
\end{proof}

\section{Bounding the excess}\label{sec:excess}
In this section we bound the total excess in terms of \OPT.
Recall that the excess contributed by a miss service on page~$p$ with path~$\lambda$ is~$\durW{W}{\lambda}$ for each level~$W$ with~$w_p\leq W< w_{\min}(\OPT)$.
This is exactly the deposit of this service into the budget~$B_W$, or~$L$ times the overshoot it claims, by definition of the algorithm.
We bound the excess in two ways:
In case \OPT has already served light pages with overshoot, we argue that preemptive services see these light pages as candidates, and therefore spend their budget completely on requests that \OPT has already paid for.
It then suffices to bound the budgets.
Otherwise, we argue that \OPT has to pay for the overshoot through delay cost, and we charge the claimed overshoot to \OPT.
To this end, we first make this case distinction precise by defining \emph{overdue} services, which are those where \OPT has already served light pages with overshoot.
\begin{definition}[overdue pages and services]
    A page~$q$ is \emph{overdue} at time~$t$ if the pending requests~$R$ on~$q$ that \OPT has served satisfy~$\sum_{r\in R} d_r(t-a_r)>w_q$.
    A miss service to~$p$ is \emph{overdue} if some page~$q$ with~$w_q\le w_p$ is overdue right before it, and \emph{timely} otherwise.
    We use these terms also for its step, and write~$\Over$ and~$\Timely$ for the sets of overdue and timely steps.
\end{definition}

Intuitively, in a timely service, \OPT has not yet served the pending overshoot on light pages.
Therefore, \OPT also has to pay the overshoot.
By definition of the algorithm, the overshoot the service claims multiplied by~$L$ is at least the excess contributed by the service.
This way, the excess of timely services can be charged to \OPT's delay cost.
The following lemma formalizes this argument.

\begin{lemma}[timely services]\label{lem:timely}
    $\sum_{\nu\in\Timely}\Exc(\nu)\le L\cdot\OPT$.
\end{lemma}
\begin{proof}
    By definition of the algorithm, every critical service with path~$\lambda$ claims~$\dur{\lambda}/L$ of delay and, if performed in step~$\nu$, contributes~$\Exc(\nu)\le\dur{\lambda}$ to the excess.
    Therefore, it suffices to show that the total amount claimed by timely services is at most \OPT.
    In particular, for every page~$q$, we measure the delay claimed from it between a visit of~$q$ or the start of the input and the last timely service that claims from it before its next visit, and show that this is at most \OPT.
    This accounts for all claims of timely services from~$q$, and possibly more.

    Fix a page~$q$ and divide the time into intervals~$(t_0, t_1]$ each ending with the next visit to~$q$.
    Since every request is eventually served, there are no requests to~$q$ after its last visit, so nothing can be claimed from~$q$ then and it suffices to consider the time intervals ending with a visit.
    Fix a time interval~$(t_0, t_1]$.
    Let~$\mu_{\rm last}$ be the last timely service in~$(t_0, t_1]$ that claims a positive amount from~$q$, and let~$t_{\rm last}\in (t_0, t_1]$ be its time.
    If there is no such service, then the total claimed from~$q$ by timely services during~$(t_0, t_1]$ is 0 and nothing has to be shown.

    Otherwise, since a claim increases~$c_q$ by at most its overshoot~$o_q$, the claims from~$q$ up to~$\mu_{\rm last}$ total at most~$\delayof_q(t_{\rm last})-w_q$, where~$\delayof_q$ refers to the requests pending right before~$\mu_{\rm last}$.
    Denote those requests by~$R$.
    Assume that \OPT pays delay cost less than~$\delayof_q(t_{\rm last})-w_q$ for the requests in~$R$.
    Then the total delay of requests in~$R$ that \OPT has served before~$\mu_{\rm last}$ is more than~$w_q$, so~$q$ is overdue right before~$\mu_{\rm last}$.
    Since~$\mu_{\rm last}$ claims a positive amount from~$q$, and a service to a page~$p$ claims only from pages of weight at most~$w_p$,~$w_q$ is at most the level of~$\mu_{\rm last}$.
    Hence~$\mu_{\rm last}$ is overdue, a contradiction.
    It follows that \OPT pays at least~$\delayof_q(t_{\rm last})-w_q$ in delay cost of~$R$, which is at least the total amount claimed from~$q$.
    The requests~$R$ are served by the visit at~$t_1$, so every request is counted for only one time interval, and the claim follows by summing over all time intervals.
\end{proof}

It remains to bound the excess of overdue services through the budgets.
As outlined, the excess can be bounded by the total deposit into the budgets.
To bound the total excess of overdue services, it therefore suffices to bound the total paid out of budget~$B_W$ plus the remaining budget.
In particular, we show that, as long as there is still an overdue service of level at most~$W$ to come, the budget~$B_W$ is spent on chargeable round trips, and the remaining budget is small.
We split this argument by \emph{periods} of \OPT, defined by the times between moves of \OPT.

First, we formally define periods.
\begin{definition}[periods]
    Let~$t_1<t_2<\dots$ be the times at which \OPT performs a walk of positive cost, and~$t_0\coloneqq-\infty$.
    \emph{Period~$j$} is the time interval~$[t_j,t_{j+1})$, with~$t_{j+1}\coloneqq\infty$ for the last period.
    \OPT's cache is constant from after its walk at~$t_j$ until before its walk at~$t_{j+1}$.
    We call this cache configuration~$\OPT_j$ (for~$j=0$, the initial cache).
\end{definition}

\begin{observation}\label{obs:servedbefore}
    Let \(r\) be a request that \OPT has already served at some time \(t\) in period \(j\).
    If \(p_r\notin\OPT_j\), then \OPT must have served the request before period \(j\) started, so \(a_r \le t_j\) and thus \(j \ge 1\).
\end{observation}
We want to argue that most of the budget is spent on chargeable round trips, except for a small remainder.
To make this precise, we introduce the early part of a period w.r.t.\ some level~$W$ and prove that, during preemptive services of the early part, the budget~$B_W$ is drained (i.e., budget is spent mostly on round trips, and not only on closing the hole), and all preemptive round trips are chargeable.

\begin{definition}[early part]
    Let~$j$ be a period and~$W<w_{\min}(\OPT_j)$ a level.
    Let~$\lightstep{1}<\dots<\lightstep{m}$ be the steps in~$\Over$ of period~$j$ of level~$\le W$, and call their critical services the \emph{light services}.
    The \emph{early part} of period~$j$ consists of all services in period~$j$ up to, but not including, the last light service.
\end{definition}

Here, the light services are the overdue services that contribute to the excess at level~$W$.

We next show that preemptive services of the early part drain their budget on chargeable round trips.
The key argument behind this statement is that, during the early part, the last light service is still to come, and some light page~$q$ is overdue.
In particular, the deadline of~$q$ has already passed.
Therefore, level-$W$ preemptive services do not run out of candidates, and instead drain their budget on requests with early deadlines.

\begin{lemma}\label{lem:early}
    Every level-$W$ preemptive service in the early part of period~$j$ performs only chargeable round trips, and ends with~$B_W<W$.
\end{lemma}
\begin{proof}
    Let~$t^{(m)}$ be the time of the last light service~$\lightstep{m}$, and let~$q$ be a page with~$w_q\le W$ that is overdue right before it.
    Such a page exists by definition of overdue services.
    By definition of \(q\) being overdue, the set \(R_q\) of pending requests to~$q$ that \OPT has served satisfy~$\sum_{r\in R_q} d_r(t^{(m)}-a_r)>w_q$.
    Because we consider the early part, we have~$W<w_{\min}(\OPT_j)$, so the page \(q\) with \(w_q \le W\) must satisfy~$q\notin\OPT_j$.
    Hence, the requests in~$R_q$ arrived no later than~$t_j$ by \Cref{obs:servedbefore} and because they are still pending right before~$\lightstep{m}$, they are pending throughout the early part.

    Consider some level-$W$ preemptive service~$\mu$ in the early part (i.e., before the last light service) and let~$p$ be the page selected for preemptive service.
    We have to show that \(B_W < W\) after \(\mu\) and that \(\mu\) is chargeable.
    Since~$q$ is eligible for preemptive service during \(\mu\) but the requests in~$R_q$ are still pending after~$\mu$, the loop of round trips of~$\mu$ must terminate because~$B_W<W$.

    Let~$R_p$ be the set of requests to~$p$ that are pending right before \(\mu\).
    Since~$p$ is chosen for preemptive service before~$q$, we have~$\deadline(p)\le\deadline(q)$ and because the pending requests to~$q$ during~$\mu$ include~$R_q$, which have total delay more than~$w_q$ at time~$t^{(m)}$, we have~$\deadline(q)\le t^{(m)}$.
    Together this gives~$\sum_{r\in R_p}d_r(t^{(m)}-a_r)\ge w_p$ by continuity of the delay cost functions.
    Again,~$p\notin\OPT_j$, as~$w_p\le W<w_{\min}(\OPT_j)$, so if \OPT has served some request in~$R_p$, the round trip satisfies~\ref{item:abandoned}.
    Otherwise, \OPT does not serve any request in~$R_p$ before~$t_{j+1}>t^{(m)}$ either, as~$p\notin\OPT_j$, so it pays at least~$\sum_{r\in R_p}d_r(t^{(m)}-a_r)\ge w_p$ of delay for them, and the round trip satisfies~\ref{item:optdelay}.
\end{proof}

We combine the above lemmas into a bound on the excess terms by bounding the total budget spent in the early part, the remaining budget after the early part, and one additional term for the final overdue light service.
Write~$C_{j,W}$ for the sum of~$w_p$ over the chargeable round trips of level~$W$ in period~$j$ to some page~$p$.
\begin{lemma}\label{lem:periodlevel}
    For every period~$j\geq 0$ and weight level~$W<w_{\min}(\OPT_j)$ it holds that
    \[
        \sum_{\ell\le m}\durW{W}{\lambda^{(\lightstep{\ell})}}\le\frac32C_{j,W}+4W+LW\,.
    \]
\end{lemma}
\begin{proof}
    By definition, each summand is at most~$LW$, and we bound the last summand this way.
    Each of the other summands~$\durW{W}{\lambda^{(\lightstep{\ell})}}$ is exactly the deposit of the~$\ell$-th light service into~$B_W$.
    Therefore, it suffices to show that
    \begin{equation}\label{ineq:early_deposits}
        \sum_{\ell< m}\durW{W}{\lambda^{(\lightstep{\ell})}} \leq \frac32C_{j,W}+4W.
    \end{equation}
    This inequality is trivial for~$m\le1$, so we assume henceforth that~$m>1$.
    To prove \eqref{ineq:early_deposits}, observe that each~$\ell$-th light service with~$\ell<m$ lies in the early part and deposits~$\durW{W}{\lambda^{(\lightstep{\ell})}}$ into the budget~$B_W$, so we can bound the sum of these terms by the total paid out of~$B_W$ during the early part (i.e., the cost of preemptive service), plus the budget remaining at the end of the early part.

    First, we examine the cost paid out of~$B_W$ during level-$W$ preemptive services in the early part.
    Every such preemptive service starts with budget at least~$4W$, ends with budget less than~$W$ (by \Cref{lem:early}) and pays no more than~$W$ to close its hole.
    Therefore, it spends more than~$2W$ on round trips, i.e., more than twice its closing cost.
    In other words, the total payment is less than~$\frac32$ times the cost of the round trips.
    By \Cref{lem:early}, every round trip in the early part is chargeable, and their costs sum to no more than~$C_{j,W}$.

    To prove \eqref{ineq:early_deposits}, it now remains only to show that the leftover budget~$B_W$ after the early part is at most~$4W$.
    Let~$\mu_{\rm last}$ be the last critical service of level~$\le W$ in the early part, which exists since the~$(m-1)$-th light service is one.
    After~$\mu_{\rm last}$, its page, which has weight~$\le W$, has missing mass~$\delta$, so~$\sum_{i:w_i\le W}(1-x_i)\ge1-\delta$.
    Therefore,~$\mu_{\rm last}$ starts a level-$W$ preemptive service as long as the budget~$B_W$ is at least~$4W$.
    In this case, since~$\mu_{\rm last}$ is in the early part, this preemptive service drains the budget to less than~$W$ by \Cref{lem:early}.
    Otherwise,~$B_W<4W$ already.
    Subsequent services in the early part have level higher than~$W$ and therefore do not increase~$B_W$.
    This concludes the proof.
\end{proof}
We have already established that the cost of chargeable round trips can be charged to \OPT.
To derive a bound on the total excess, we additionally need to bound the extra summand~$LW$ in \Cref{lem:periodlevel}.
To this end, we note that~$W<w_{\min}(\OPT_j)$, and we can bound the movements of \OPT between periods from below by~$w_{\min}(\OPT_j)$.
\begin{lemma}[lightest weights]\label{lem:lightest}
    $\sum_{j\ge1}w_{\min}(\OPT_j)\le\OPT$.
\end{lemma}
\begin{proof}
    Since we assume \OPT is lazy, the last page fetched in the walk at~$t_j$ lies in~$\OPT_j$, so its weight is at least~$w_{\min}(\OPT_j)$.
    Summing over the distinct walks at~$t_j$ for all~$j\geq 1$ gives the claimed bound.
\end{proof}

Finally, we can combine the lemmas from this section to derive a bound on the total excess.
\begin{proposition}[total excess]\label{prop:excess}
    $\Exc=\O(L)\cdot \OPT+\O(L\sum_i w_i)$.
\end{proposition}
\begin{proof}
    Every step~$\nu\in\Miss$ is either in~$\Over$ or in~$\Timely$, so~$\Exc=\sum_{\nu\in\Over}\Exc(\nu)+\sum_{\nu\in\Timely}\Exc(\nu)$.
    By \Cref{lem:timely}, timely services contribute at most~$L\cdot\OPT$ to the excess.
    Every step~$\nu\in\Over$ in period~$j$ contributes to~$\Exc$ only at levels~$W$ with~$w_{p^{(\nu)}}\le W<w_{\min}(\OPT_j)$, and for each such level,~$\nu=\lightstep{\ell}$ for some~$\ell$.
    If such a level exists, then the page~$q$ that is overdue right before~$\nu$ satisfies~$w_q\le w_{p^{(\nu)}}<w_{\min}(\OPT_j)$, so~$q\notin\OPT_j$, and as \OPT has served a pending request to~$q$, \Cref{obs:servedbefore} gives~$j\ge1$.
    Hence,
    \begin{align*}
        \sum_{\nu\in\Over}\Exc(\nu)
        \leq \sum_{j\geq 1}\,\sum_{W<w_{\min}(\OPT_j)}\,\sum_{\ell} \durW{W}{\lambda^{(\lightstep{\ell})}}
        \,.
    \end{align*}
    Applying \Cref{lem:periodlevel} gives
    \[
        \sum_{\nu\in\Over}\Exc(\nu)
        \leq \sum_{j\geq 1} \,\sum_{W<w_{\min}(\OPT_j)} \Big(\frac32 C_{j,W}+(L+4)W\Big).
    \]
    \Cref{lem:chargeable} gives the bound~$\sum_{j,W}C_{j,W}\le\OPT+w(\OPT_0)$.
    For the terms~$(L+4)W$, observe that for every~$j$, since weight levels are distinct powers of two,
    \[
        \sum_{W<w_{\min}(\OPT_j)} W < w_{\min}(\OPT_j).
    \]
    Applying \Cref{lem:lightest} yields
    \[
        \sum_{j\geq 0} \sum_{W<w_{\min}(\OPT_j)} W \leq \OPT + w(\OPT_0)\,.
    \]
    Combining the above inequalities gives the claimed bound
    \[
        \Exc \leq L\cdot\OPT+\frac 3 2 (\OPT+w(\OPT_0))+(L+4)(\OPT+w(\OPT_0)) = \O(L)\cdot \OPT+\O(L\sum_i w_i).\qedhere
    \]
\end{proof}

\section{Proof of the main theorem}\label{sec:main}

Combining the results from the above sections gives that \PMD has fractional fetch and delay cost
\[
    \O(\log k)\cdot \OPT+\O(\log k \sum_i w_i).
\]
It remains to transform this online into a feasible randomized integral algorithm.
We do so by combining the online rounding theorem from \cite{Bansal07PrimalDual} with a rounding map similar to the one from \cite{Bubeck18kServerMultiscale}.
The proof can be found in Appendix \ref{app:rounding}.
\begin{restatable}[rounding; proof in \Cref{app:rounding}]{proposition}{roundingtheorem}\label{prop:rounding}
    Let the weights be powers of two, and let~$C_{\fetch}$ and~$C_{\delayof}$ be the fractional fetch and delay cost of \PMD (\Cref{def:fracfetch,def:fracdelay}).
    Then the states of \PMD can be converted online into a randomized integral algorithm that serves every request with probability~$1$, has expected fetch cost~$\O(C_{\fetch}+\sum_iw_i)$, and expected delay cost at most~$3C_{\delayof}$.
\end{restatable}
\begin{proof}[Proof sketch]
    First, we apply the transformation~$z_i\coloneqq\max\{0,2\frac{1-x_i}{1-\delta}-1\}$ to the fractional antimass states to get new states describing the cached masses.
    This transformation maps antimass~$\delta$ to cached mass~$0$, so that every page with antimass~$\delta$ is fully cached.
    Moreover, the total cached mass is at most~$k$, and the transformation is~$3$-Lipschitz (\Cref{lem:map}).
    It follows that the transformed states are valid fractional cache states that respect the serving times~$\tau_r$ we planned.

    The online rounding of \citeauthor{Bansal07PrimalDual} then transforms the fractional cache states into random cache choices that contain every page with probability equal to its cached mass and realize the fractional fetch cost up to a constant factor (\Cref{cor:rounding}).
    In particular, a request~$r$ on~$p$ is served at the latest at~$\tau_r$, when~$x_p=\delta$ and~$p$ is cached with probability~$1$.
    When~$r$ arrives,~$p$ is not cached with probability at most~$1-z_p\le3(x_p-\delta)=3\alpha_r$, so the expected delay of~$r$ is at most~$3\alpha_rd_r(\tau_r-a_r)$.
    This justifies that the fractional delay cost weights the delay by~$\alpha_r$.
    The full proof is in \Cref{app:rounding}.
\end{proof}
Applying this rounding scheme gives the desired~$\O(\log k)$-competitive randomized algorithm.

\begin{proof}[\textnormal{\textbf{Proof of \Cref{thm:main}}}]
    First, we round the weights up to powers of two.
    This at most doubles \OPT and~$\sum_iw_i$, and does not decrease the cost of any algorithm.
    By \Cref{prop:tele} and \Cref{prop:excess}, it holds on every instance that
    \[
        C_{\fetch}^{\crit}\mathrel{\le}\O(L)\cdot \OPT + \O(L\sum_iw_i) = \O(\log k)\cdot \OPT+\O(\log k\sum_iw_i).
    \]
    \Cref{prop:cost} with~$\rho=\O(\log k)$ and~$c=\O(\log k\sum_iw_i)$ then gives
    \[
        C_{\fetch}+C_{\delayof}\le\O(\log k)\cdot\OPT+\O(\log k\sum_iw_i).
    \]
    By \Cref{prop:rounding}, the states of \PMD can be converted online into a randomized integral algorithm that serves every request with probability~$1$ and has expected cost at most
    \[
        \O\Big(C_{\fetch}+\sum_iw_i\Big)+3C_{\delayof}\le\O(\log k)\cdot\OPT+\O\Big(\log k\sum_iw_i\Big).
    \]
    The additive term depends only on~$k$ and the weights, not on the requests, so this algorithm is~$\O(\log k)$-competitive.
\end{proof}

\section*{Disclosure of AI usage}
An initial draft of this paper was written by Claude Opus 5.5 based on two proofs written by the authors.
The first showed the competitive ratio~$\O(\log k)$ for an algorithm based on mirror descent for paging with deadlines in two special cases: when there are only two weight classes, and when the loaded page is the heaviest page.
The second contained a fully potential-based analysis of the result from \cite{Azar21OSD}, specialized to paging with deadlines, giving a potential that amortizes the cost difference between paging with and without delay, and that was then used in the initial draft of this paper.
Claude Opus extended our results to general weights and delays, in particular by introducing weight classes explicitly and depositing preemptive service budgets based on these weight classes.
The Lean formalization was written by Claude Opus.
The authors guided the design and verified the correctness of the Lean model for randomized weighted paging with delay.

This paper was written entirely by the authors, based on the initial draft provided by Claude Opus.
The authors independently checked all resulting proofs and take full responsibility for the content and accuracy of the paper, including any errors that may remain.

\printbibliography

\appendix

\crefalias{section}{appendix}

\section{Rounding}\label{app:rounding}

For classical weighted paging, \citeauthor{Bansal07PrimalDual} \cite{Bansal07PrimalDual} round a fractional algorithm online:
they maintain a random cache so that the probability that a page is cached is exactly its fractional cached mass.
Their rounding procedure changes the algorithm's cost by only a constant factor.

In this section, we show that their rounding procedure can be applied to the fractional states of \PMD after applying a fixed transformation to the fractional states.
This yields the desired randomized algorithm with only a constant factor increase in cost.

\begin{theorem}[{\cite[Theorem~4.1]{Bansal07PrimalDual}}]\label{thm:bbn}
    Let the weights~$\omega_i$ be such that~$\omega_{i+1}/\omega_i\ge2$ for~$1\le i\le K-1$.
    There is a subclass~$\mathcal D$ of distributions on cache states, along with a map~$T$ from~$(\mathcal D\times\mathcal P)\to\mathcal D$ with the following property:
    Given any two distributions on pages~$y,\tilde y\in\mathcal P$, and given any~$Q\in\mathcal D$ satisfying~$\Pi(Q)=y$, we can obtain another distribution~$\tilde Q=T(Q,\tilde y)$ such that~$\Pi(\tilde Q)=\tilde y$,~$\tilde Q\in\mathcal D$ and~$\mathrm{cost}(Q,\tilde Q)\le5\,\mathrm{cost}_f(y,\tilde y)$.
\end{theorem}

In \Cref{thm:bbn}, the notation is as follows.
\begin{description}[font=\normalfont, leftmargin=2.9cm, labelwidth=2.6cm, labelsep=0.3cm, align=left, itemsep=0pt]
    \item[$\omega_1<\dots<\omega_K$] the distinct page weights.
    \item[cache state] a set of~$k$ pages.
    \item[$Q,\tilde Q$] probability distributions on cache states, i.e.,~$Q$ gives for every cache state~$C$ the probability that the random cache is~$C$.
    \item[$\mathcal P$] the vectors~$y\in[0,1]^n$ with~$\sum_iy_i=k$, which \cite{Bansal07PrimalDual} call \emph{distributions on pages};~$y_i$ is the cached mass of page~$i$.
    \item[$\Pi(Q)$] the \emph{marginals} of~$Q$:~$\Pi(Q)_i$ is the probability that page~$i$ is in a random cache with distribution~$Q$.
    \item[$\mathcal D$] a subclass of the distributions on cache states, constructed in \cite{Bansal07PrimalDual}, whose definition we do not need.
    \item[$\mathrm{cost}(Q,\tilde Q)$] the minimum expected cost of a \emph{transition from~$Q$ to~$\tilde Q$}, i.e., of a rule that says, for every current cache state, with which probability to move to each cache state, such that a cache with distribution~$Q$ ends up with distribution~$\tilde Q$.
          Fetching and evicting page~$i$ each cost~$\frac{w_i}2$ (the convention of \cite{Bansal07PrimalDual}).
    \item[$\mathrm{cost}_f(y,\tilde y)$]~$\sum_i\frac{w_i}2|y_i-\tilde y_i|$, the cost of moving fractional mass from~$y$ to~$\tilde y$ in the same convention.
    \item[$T$] the update rule: it picks the next distribution~$T(Q,\tilde y)$ from the current one and the new marginals.
\end{description}

In the notation of this paper, \Cref{thm:bbn} gives the following.
\begin{corollary}[online rounding]\label{cor:rounding}
    Let the weights be powers of two, and let~$y^{(0)},\dots,y^{(N)}\in[0,1]^n$ with~$\sum_iy^{(\nu)}_i=k$ be given online, where~$y^{(0)}$ is the indicator vector of~$\OPT_0$.
    Then one can choose random sets~$\mathcal C^{(0)},\dots,\mathcal C^{(N)}$ of~$k$ pages online, i.e., each~$\mathcal C^{(\nu)}$ depends only on~$y^{(0)},\dots,y^{(\nu)}$ and random choices made up to then.
    They satisfy~$\mathcal C^{(0)}=\OPT_0$,~$\Pr[i\in\mathcal C^{(\nu)}]=y^{(\nu)}_i$ for all~$\nu$ and all pages~$i$, and
    \[
        \mathbb E\Big[\sum_{i\in\mathcal C^{(\nu+1)}\setminus\mathcal C^{(\nu)}}w_i\Big]\le5\sum_iw_i\big|y^{(\nu+1)}_i-y^{(\nu)}_i\big|\qquad\text{for all }\nu<N.
    \]
\end{corollary}
\begin{proof}
    The assumptions of \Cref{thm:bbn} hold: as the distinct weights are powers of two,~$\omega_{i+1}/\omega_i\ge2$, and every~$y^{(\nu)}$ lies in~$\mathcal P$.
    By the proof of \cite[Theorem~4.1]{Bansal07PrimalDual},~$\mathcal D$ contains a distribution~$Q^{(0)}$ with~$\Pi(Q^{(0)})=y^{(0)}$.
    As~$y^{(0)}$ is the indicator vector of~$\OPT_0$,~$Q^{(0)}$ gives probability~$1$ to the cache state~$\OPT_0$, so we can let~$\mathcal C^{(0)}\coloneqq\OPT_0$.
    Given~$\mathcal C^{(\nu)}$ with distribution~$Q^{(\nu)}\in\mathcal D$, let
    \[
        Q^{(\nu+1)}\coloneqq T(Q^{(\nu)},y^{(\nu+1)}),
    \]
    and obtain~$\mathcal C^{(\nu+1)}$ by applying a cheapest transition from~$Q^{(\nu)}$ to~$Q^{(\nu+1)}$ to the current cache~$\mathcal C^{(\nu)}$.
    As~$\mathcal C^{(\nu)}$ has distribution~$Q^{(\nu)}$, this gives, at expected cost~$\mathrm{cost}(Q^{(\nu)},Q^{(\nu+1)})$, a cache~$\mathcal C^{(\nu+1)}$ with distribution~$Q^{(\nu+1)}$.
    By \Cref{thm:bbn},~$Q^{(\nu+1)}\in\mathcal D$, so that this can be repeated, and~$\Pr[i\in\mathcal C^{(\nu+1)}]=\Pi(Q^{(\nu+1)})_i=y^{(\nu+1)}_i$.
    As~$Q^{(\nu+1)}$ depends only on~$y^{(0)},\dots,y^{(\nu+1)}$, the sets are chosen online.

    Finally,~$\mathrm{cost}(Q^{(\nu)},Q^{(\nu+1)})$ charges~$\frac{w_i}2$ for every fetched page~$i\in\mathcal C^{(\nu+1)}\setminus\mathcal C^{(\nu)}$, so by \Cref{thm:bbn},
    \[
        \mathbb E\Big[\sum_{i\in\mathcal C^{(\nu+1)}\setminus\mathcal C^{(\nu)}}w_i\Big]
        \le2\,\mathrm{cost}(Q^{(\nu)},Q^{(\nu+1)})
        \le10\,\mathrm{cost}_f(y^{(\nu)},y^{(\nu+1)})
        =5\sum_iw_i\big|y^{(\nu+1)}_i-y^{(\nu)}_i\big|.\qedhere
    \]
\end{proof}

The cached masses~$1-x_i$ of our fractional algorithm do not fit \Cref{cor:rounding} directly: a page with~$x_i=\delta$ is not fully cached, and the total cached mass may be less than~$k$.
For the first issue, we stretch the cached mass by~$\frac1{1-\delta}$, which makes such pages fully cached but may overshoot~$k$ by~$\frac12$, and remove the overshoot by doubling and subtracting~$1$.

\begin{lemma}\label{lem:map}
    Let~$x\in[\delta,1]^n$ with~$\sum_ix_i\ge n-k$.
    Define~$u_i\coloneqq(1-x_i)/(1-\delta)$ and~$z_i\coloneqq\max\{0,2u_i-1\}$.
    Then~$z\in[0,1]^n$,~$z_i=1$ if~$x_i=\delta$,~$\sum_iz_i\le k$, and each~$z_i$ is a nonincreasing,~$3$-Lipschitz function of~$x_i$.
\end{lemma}
\begin{proof}
    As~$x_i\in[\delta,1]$, we have~$u_i\in[0,1]$ and hence~$z_i\in[0,1]$.
    If~$x_i=\delta$, then
    \[
        z_i = \max\{0, 2\cdot (1-\delta)/(1-\delta)-1\}=1.
    \]
    As a function of~$x_i$,~$z_i$ is piecewise linear with slopes~$0$ and~$-\frac2{1-\delta}=-2-\frac1k\ge-3$, hence nonincreasing and~$3$-Lipschitz.
    We have~$\sum_iu_i\le\frac{k}{1-\delta}=k \cdot \frac{2k+1}{2k}=k+\frac12$.
    If at most~$k$ pages have~$u_i>\frac12$, then only those pages have~$z_i > 0$ and hence~$\sum_i z_i \leq k$.
    Otherwise,~$\sum_iz_i\le2(k+\frac12)-(k+1)=k$.
\end{proof}

\begin{remark}
    A rounding map as in \Cref{lem:map} can also be defined for values other than $\delta=\frac 1 {2k+1}$.
    To apply it to our problem, it has to map antimass $\leq \delta$ to cached mass $1$, so that the page is cached almost surely.
    It is immediate from the cache constraint of paging that our approach requires $\delta < \frac 1 {k+1}$:
    Otherwise, if $k+1$ nodes have antimass $\frac 1 {k+1}$ and the remaining $n-k-1$ have antimass $1$, the fractional cache constraint $\sum_i x_i\ge n-k$ is met, but the rounded state has to cache $k+1$ pages at once.

    The fractional analysis of \PMD loses a factor of order~$L=\log\frac1\delta$, so it favors a large~$\delta$.
    Hence, the lower bound $\Omega(\log k)$ on the competitive ratio is inherent to our analysis.
\end{remark}

It remains to transform the rounded states into a randomized algorithm and to analyze the expected cost in terms of the fractional fetch and delay cost.
Here, we first transform the fractional sequence to be lazy, so that the conditions of \Cref{cor:rounding} are met: we keep a state~$y\ge z$ of total mass exactly~$k$, and decrease~$y_i$ only when its mass is needed on other pages.

\roundingtheorem*

\begin{proof}[Proof of \Cref{prop:rounding}]
    By \Cref{lem:inv}, every state~$x^{(\nu)}$ of \PMD satisfies~\(
    x^{(\nu)}\in[\delta,1]^n
    \)
    and has total antimass~$\sum_ix^{(\nu)}_i\ge n-k$,
    so we can apply \Cref{lem:map} to it.
    Let~$z^{(\nu)}\coloneqq z(x^{(\nu)})$ be as in \Cref{lem:map}.
    Before we can apply \Cref{cor:rounding}, we first need to transform the states into a lazy strategy so that the total cached mass is exactly~$k$.

    To this end, let~$y^{(0)}\coloneqq z^{(0)}$, and obtain~$y^{(\nu+1)}$ from~$z^{(\nu+1)}$ by increasing coordinates that are smaller than in~$y^{(\nu)}$ until~$\sum_iy_i^{(\nu+1)}=k$ (i.e., mass is only removed from coordinates when the total mass would otherwise exceed~$k$).
    This transformation is online, since~$y^{(\nu+1)}$ depends only on~$y^{(\nu)}$ and~$x^{(\nu+1)}$.
    The vectors~$y^{(\nu)}$ have total mass~$k$ and lie in~$[0,1]^n$.
    We let the initial state of \PMD be~$x^{(0)}$ such that~$x_i=\delta$ for~$i\in\OPT_0$,~$x_i=1-k\delta=(k+1)\delta$ for some page~$i\notin\OPT_0$, and~$x_i=1$ for all other pages.
    Applying \Cref{lem:map} gives that~$y^{(0)}=z^{(0)}$ is the indicator vector of~$\OPT_0$.
    Therefore, \Cref{cor:rounding} can be applied and gives an online randomized integral algorithm that starts with the cache~$\OPT_0$ and whose cache~$\mathcal C^{(\nu)}$ at the state~$x^{(\nu)}$ contains each page~$i$ with probability~$y^{(\nu)}_i$.

    We now bound the expected fetch cost of this algorithm by~$30C_{\fetch}+5\sum_iw_i$.
    By \Cref{cor:rounding}, it is at most~$5\sum_{\nu<N}\sum_iw_i\big|y^{(\nu+1)}_i-y^{(\nu)}_i\big|$.
    As every~$y_i$ starts at most~$1$ and stays nonnegative, its total decrease is at most its total increase plus~$1$, so
    \[
        5\sum_{\nu<N}\sum_iw_i\big|y^{(\nu+1)}_i-y^{(\nu)}_i\big|\le10\sum_{\nu<N}\sum_iw_i\big(y^{(\nu+1)}_i-y^{(\nu)}_i\big)_++5\sum_iw_i.
    \]
    By the definition of~$C_{\fetch}$ (\Cref{def:fracfetch}), it remains to show that~$y_i$ increases at most~$3$ times as much as~$x_i$ decreases.
    As~$y_i$ is raised only up to~$z^{(\nu+1)}_i$ and~$y^{(\nu)}_i\ge z^{(\nu)}_i$, it increases at most as much as~$z_i$.
    As~$z_i$ is a nonincreasing,~$3$-Lipschitz function of~$x_i$ (\Cref{lem:map}),~$z_i$ in turn increases at most~$3$ times as much as~$x_i$ decreases.
    Hence, for all~$\nu<N$ and all pages~$i$,
    \[
        \big(y^{(\nu+1)}_i-y^{(\nu)}_i\big)_+\le\big(z^{(\nu+1)}_i-z^{(\nu)}_i\big)_+\le3\big(x^{(\nu)}_i-x^{(\nu+1)}_i\big)_+.
    \]

    Finally, we bound the expected delay cost of every request~$r$ by~$3\alpha_r d_r(\tau_r-a_r)$.
    Let~$r$ be a request to~$p$,~$x$ be the fractional state right before time~$a_r$, so that~$\alpha_r= x_p-\delta$ (\Cref{def:fracdelay}), and let~$\mathcal C$ be the random cache right before~$a_r$.
    If~$p\in\mathcal C$, then~$r$ is served at time~$a_r$ with no delay cost.
    Otherwise,~$r$ is served no later than at~$\tau_r$, where~$x_p=\delta$, so~$z_p=1$ and~$p$ is cached with probability~$y_p\ge z_p=1$.
    The delay cost of~$r$ then is at most~$d_r(\tau_r-a_r)$.
    Therefore, it suffices to show~$\Pr[p\notin\mathcal C]\le3\alpha_r$.
    By \Cref{lem:map},~$z_p$ is a~$3$-Lipschitz function of~$x_p$ with value~$1$ at~$x_p=\delta$, hence
    \[
        \Pr[p\notin\mathcal C]=1-y_p\overset{y_p\geq z_p}{\le}1-z_p( x)\le3( x_p-\delta)=3\alpha_r.
    \]
    Summed over all requests, the expected delay cost is at most~$3C_{\delayof}$.
\end{proof}

\section{Why simpler rules for when to serve fail}\label{app:obstacles}
The following examples motivate why we use \eqref{eq:trigger} to determine when to perform critical service.
For convenience, we recall that~\(p\) becomes critical at the last moment before
\[
    \big(\delayof_p(t)-c_p-w_p\big)_+ > 0
    \qquad\text{and}\qquad
    \sum_{q\colon w_q\le w_p}\big(\delayof_q(t)-c_q-w_q\big)_+\ge\dur{\lpath{x}{p}}/L.
\]
\begin{example}[Dependence on~\(\dur{\lpath{x}{p}}\)]\label{ex:room}
    This example shows that the rule for performing critical service must not depend only on the weight of the page and its delay.

    Let~$k=1$, and let there be a heavy and a light page~$h$ and~$\ell$ with weights~$W$ and~$1$, respectively.
    The pages are requested in alternating order until some time~\(T\).
    The heavy page~$h$ is requested at every even time~$2t<T$ and its delay~$d(\tau)=M\tau$ grows linearly with a very fast rate~\(M\gg W\).
    Page~$\ell$ is requested at every odd time~$2t+1<T$ and the delay only grows to some~\(\gamma\) chosen below, remains there until time~\(T\) and then grows at the fast rate~\(M\), i.e.,~$d_r(\tau)=\gamma\min\{\tau,1\}$ for~$a_r+\tau\le T$, and~\(d_r(\tau) = \gamma + M(a_r+\tau-T)\) afterwards.

    \OPT caches~$h$ immediately at cost~\(W\), keeps it until time~$T$ (incurring no cost from the requests to~\(h\)) and then serves all requests to~$\ell$, paying at most~$\gamma$ per request to~$\ell$ plus~$1$ to fetch it.
    In total,~\(\OPT = \gamma\cdot\tfrac T2 + \O(W)\).

    Consider any rule that serves~$\ell$ once its pending delay exceeds a threshold~$\kappa$ that depends only on~$w_\ell$.
    Choosing~$\gamma>\kappa$ forces such a rule to fetch~$\ell$ within one time unit after every request to~$\ell$, and the next request to~$h$ forces it to fetch~$h$ back at cost~\(W\) or to pay delay of at least~$M>W$.
    Either way, the cost is~$\Omega(W)$ per round, so the competitive ratio is~$\Omega(W/\gamma)$.
    Thus, any rule depending only on the delay and weight of~\(\ell\) is not sufficient.
\end{example}

\begin{example}[Dependence on lighter pages]\label{ex:shared}
    This example shows that the rule for performing critical service must not depend only on the page to be turned critical.

    Suppose all light pages are fractionally evicted and the cache is filled with heavy pages of weight~$W$, while \OPT caches a light page of weight~\(1\).
    Now several light pages~$\ell_1,\ell_2,\dots,\ell_m$ receive requests with the same delay function~\(d(\tau)=\tau\) at the same time, and \OPT serves them immediately through its light slot, which costs~\(m\) in total.

    A rule that serves each light page~$\ell_i$ once its own pending delay exceeds~$\dur{\lpath{x}{\ell_i}}/L$ lets all~$\ell_i$ accumulate delay in parallel until their delay is~$W$ each, resulting in a cost of~\((W+1)m\).

    \PMD solves this by letting all lighter pages contribute to when a page becomes critical.
\end{example}

\section{The mirror descent path in closed form}\label{app:mdpath}
In \Cref{sec:alg}, we define the path~$\lambda=\lpath{x}{p}$ of a critical service by a closed formula instead of by the differential inclusion \eqref{eq:md}, and our analysis uses only this closed form.
This section shows that~$\lambda$ is indeed the mirror descent path of \cite[Section~2.2]{Bubeck18kServerMultiscale}, traversed at a different speed.
The idea is that \eqref{eq:md}, written componentwise, lets the antimass~$y_i$ of every page~$i\ne p$ with~$y_i<1$ grow at the relative rate~$(\log y_i)'=\mu(t)/w_i$, where the factor~$\mu(t)$ is the same for all these pages.
Measuring progress by~$s=\int\mu(t)\dd t$ instead of by~$t$ removes this factor and gives~$y_i=x_ie^{s/w_i}$, the closed form of~$\lambda$.
We first write \eqref{eq:md} componentwise, then determine~$\mu$, and finally verify that~$\lambda$, traversed at a suitable speed, satisfies \eqref{eq:md}.

Throughout, let~$x\in S\cap[\delta,1]^n$ be the start state, with~$x_p>\delta$, and let~$\lambda=\lpath{x}{p}$.

\paragraph{The differential inclusion for each page}
Recall that the mirror descent path is a path~$y(t)$ in~$S=\{y\in[0,1]^n:\sum_iy_i=n-k\}$ with~$y(0)=x$ that satisfies
\[
    y'(t)\in\nabla^2\Phi(y(t))^{-1}\cdot\big(-\mathbf 1_p-N_S(y(t))\big),
    \tag{\ref{eq:md}}
\]
where~$N_S(y)$ is the normal cone of~$S$ at~$y$.
The mirror map~$\Phi(y)=\sum_iw_iy_i\log y_i$ has the diagonal Hessian~$\nabla^2\Phi(y)$ with entries~$w_i/y_i$.
For~$y\in S$ with~$y>0$, the only constraints of~$S$ that can be tight are~$\sum_iy_i=n-k$ and~$y_i\le1$, so the normal cone is
\[
    N_S(y)=\Big\{\sum_{i:y_i=1}\eta_i\mathbf 1_i-\mu\mathbf 1:\ \mu\in\mathbb R,\ \eta_i\ge0\Big\}.
\]
Here,~$\mathbf 1$ is the vector of all ones and~$\mathbf 1_i$ is the vector with a~$1$ in position~$i$ and~$0$ elsewhere.
Hence, a path~$y(t)\in S$ with~$y(t)>0$ satisfies \eqref{eq:md} at time~$t$ if and only if there are~$\mu(t)\in\mathbb R$ and~$\eta_i(t)\ge0$, with~$\eta_i(t)=0$ unless~$y_i(t)=1$, such that
\begin{equation}\label{eq:md_componentwise}
    y_i'(t)=\frac{y_i(t)}{w_i}\big(\mu(t)-\mathbf 1_p(i)-\eta_i(t)\big)\qquad\text{for all pages }i.
\end{equation}

\paragraph{Where the closed form comes from}
Summing \eqref{eq:md_componentwise} over all pages and using~$\sum_iy_i'(t)=0$, which keeps~$y$ in~$S$, gives
\[
    \mu(t)\sum_i\frac{y_i(t)}{w_i}=\frac{y_p(t)}{w_p}+\sum_{i:y_i(t)=1}\frac{\eta_i(t)}{w_i},
\]
so~$\mu(t)>0$.
Every page~$i\ne p$ with~$y_i(t)<1$ has~$\eta_i(t)=0$ and, by \eqref{eq:md_componentwise}, grows at the relative rate~$(\log y_i)'(t)=\mu(t)/w_i$.
A page~$i\ne p$ with~$y_i(t)=1$ therefore stays at antimass~$1$: it cannot exceed~$1$, and as soon as it dropped below~$1$, it would grow again.
So~$y_i'(t)=0$, i.e.,~$\eta_i(t)=\mu(t)$.
Substituting this into the sum determines the factor as~$\mu(t)=\hat\mu(y(t))$ \cite[(2.8)]{Bubeck18kServerMultiscale}, where for a state~$z>0$ we let
\[
    A(z)\coloneqq\{i\ne p:z_i<1\}\cup\{p\}
    \qquad\text{and}\qquad
    \hat\mu(z)\coloneqq\frac{z_p/w_p}{\sum_{i\in A(z)}z_i/w_i}.
\]
With~$s(t)\coloneqq\int_0^t\mu(\tau)\dd\tau$, integrating~$(\log y_i)'=\mu/w_i$ gives~$y_i(t)=x_ie^{s(t)/w_i}$ until~$y_i$ reaches~$1$, i.e.,~$y_i(t)=\lambda(s(t))_i$.

\paragraph{$\lambda$ satisfies the differential inclusion}
For~$s\in[0,\dur{\lambda}]$, we have~$\lambda(s)\ge\delta>0$, since~$\lambda(s)_i\ge x_i\ge\delta$ for~$i\ne p$ and~$\lambda(s)_p\ge\delta$ by the choice of~$\dur{\lambda}$.
Hence,~$\sigma\mapsto\hat\mu(\lambda(\sigma))$ is positive on~$[0,\dur{\lambda}]$, and it is piecewise continuous, with jumps only where a page reaches antimass~$1$.
Let~$s(t)$ be the inverse of the increasing function~$s\mapsto\int_0^s\frac{\dd\sigma}{\hat\mu(\lambda(\sigma))}$, so that~$s'(t)=\hat\mu(\lambda(s(t)))$, and let~$y(t)\coloneqq\lambda(s(t))$.
We show \eqref{eq:md_componentwise} with~$\mu(t)\coloneqq\hat\mu(y(t))>0$, and~$\eta_i(t)\coloneqq\mu(t)$ for the pages~$i\ne p$ with~$y_i(t)=1$, and~$\eta_i(t)\coloneqq0$ for all other pages.
Let~$A=A(y(t))$ be the set of pages whose mass is still changing, and let~$t$ be a time at which no page reaches antimass~$1$.
\begin{itemize}
    \item For~$i\in A\setminus\{p\}$, we have~$y_i=x_ie^{s/w_i}$, so~$y_i'(t)=s'(t)\,y_i(t)/w_i=\mu(t)\,y_i(t)/w_i$.
    \item For~$i\notin A$, i.e.~$y_i(t)=1$, the page stays at antimass~$1$, so~$y_i'(t)=0=\frac{y_i(t)}{w_i}\big(\mu(t)-\eta_i(t)\big)$.
    \item For~$p$, we have~$y_p=n-k-\sum_{i\ne p}y_i$, so by the two cases above and the definition of~$\hat\mu$,
          \[
              y_p'(t)
              =-\mu(t)\sum_{i\in A\setminus\{p\}}\frac{y_i(t)}{w_i}
              =-\hat \mu(y(t))\sum_{i\in A}\frac{y_i(t)}{w_i}+\mu(t)\frac{y_p(t)}{w_p}
              =-\frac{y_p(t)}{w_p}+\mu(t)\frac{y_p(t)}{w_p}
              =\frac{y_p(t)}{w_p}\big(\mu(t)-1\big).
          \]
\end{itemize}
So~$y$ satisfies \eqref{eq:md_componentwise}, and hence \eqref{eq:md}, at almost all times.
The mirror descent path~$y$ passes through the same states as~$\lambda$, in the same order, and reaches~$y_p=\delta$ exactly when~$s(t)=\dur{\lambda}$.

\end{document}